\documentclass[journal]{IEEEtran}

\usepackage{amsmath,amssymb,amsfonts,amsthm}
\usepackage{algorithmic}
\usepackage{algorithm}
\usepackage{array}
\usepackage[caption=false,font=normalsize,labelfont=sf,textfont=sf]{subfig}
\usepackage{textcomp}
\usepackage{stfloats}
\usepackage{url}
\usepackage{verbatim}
\usepackage{graphicx}
\usepackage{cite}

\usepackage{placeins}
\usepackage{xcolor}
\usepackage{bm}
\usepackage{booktabs}
\usepackage{tikz}
\usetikzlibrary{shapes,arrows,positioning,calc}
\usepackage{soul}

\newtheorem{theorem}{Theorem}
\newtheorem{lemma}{Lemma}
\newtheorem{corollary}{Corollary}
\newtheorem{proposition}{Proposition}
\newtheorem{definition}{Definition}
\newtheorem{remark}{Remark}

\newtheorem{problem}{Problem}

\begin{document}

\title{Robust Safety Filter Synthesis for Quaternion Attitude Dynamics via LMI-Based Ellipsoidal Invariant Sets}

\author{Reza Pordal, Alireza Sharifi, and Ali BaniAsad

\thanks{Manuscript received May 2026. This work received no specific
grant from any funding agency in the public, commercial, or
not-for-profit sectors.}
\thanks{R. Pordal is with the Department of Aerospace Engineering,
Sharif University of Technology, Tehran 11155-8639, Iran
(e-mail: reza.pordal77@sharif.edu).}
\thanks{A. Sharifi is with the Department of Aerospace Engineering,
Sharif University of Technology, Tehran 11155-8639, Iran
(e-mail: ar.sharifi@sharif.edu). Corresponding author.}
\thanks{A. BaniAsad is with the Department of Aerospace Engineering,
Sharif University of Technology, Tehran 11155-8639, Iran
(e-mail: ali.baniasad1999@sharif.edu).}
}

\markboth{}%
{Pordal \MakeLowercase{\textit{et al.}}: Robust Safety Filter for Quaternion Attitude Dynamics}


\maketitle

\begin{abstract} 
    We present a safety filter to guarantee constraint satisfaction on the rotation angle in the presence of disturbances. An LMI-based framework simultaneously synthesizes a maximal ellipsoidal robust controlled invariant (RCI) set and its associated state-feedback backup control law by solving a single convex semidefinite program, subject to state and input constraints. To extend this framework to nonlinear quaternion attitude dynamics, we derive exact closed-form sector bounds on the quaternion kinematic nonlinearity and analytically embed them into the LMI via the S-procedure. A smooth mixing law intervenes only as the state approaches the RCI boundary, preserving nominal performance during safe operation. This work is motivated by hierarchical aerial control architectures, where outer-loop commands can generate attitude references that drive the inner-loop attitude state unstable, a cascade failure mode that endangers the entire system. Quadrotor simulations with hierarchical controller structures under bounded disturbances confirm constraint satisfaction across three scenarios specifically designed to stress-test the cascade failure mode: set-point tracking with small initial errors, set-point tracking with large initial position errors that saturate the outer loop, and high-frequency circular trajectory following that persistently excites the inner-loop attitude dynamics.
\end{abstract}


\section{Introduction}
\label{sec:introduction}

\begin{figure}[!t]
    \centering
    \includegraphics[width=\columnwidth]{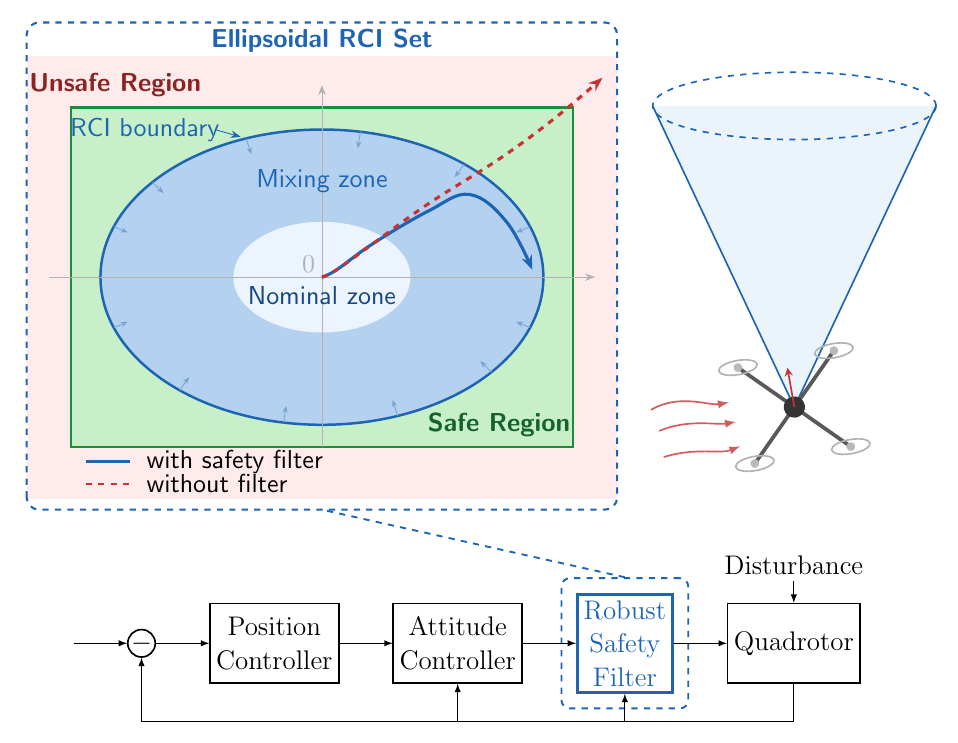}
    \caption{Overview of the proposed framework. An LMI-based ellipsoidal RCI set and backup controller are synthesised for quaternion attitude dynamics using exact sector bounds via the S-procedure. A smooth blending law enforces safety with minimal intervention, validated on three quadrotor scenarios under bounded disturbances.}
    \label{fig:graphical_abstract}
\end{figure}

\IEEEPARstart{H}{ierarchical} control architectures are widely deployed in autonomous aerial
systems, where an outer position loop generates attitude reference commands for an inner attitude
loop without direct knowledge of the inner loop's safety constraints. Due to significant initial displacement or aggressive reference trajectories, the outer loop
can issue attitude commands that drive the inner-loop state beyond the region where stability is
formally certified. For attitude dynamics governed by quaternion kinematics, this corresponds to
pushing the rotation angle and angular rates into configurations where kinematic nonlinearities
become significant, causing linearization-based inner-loop controllers to lose their stability
guarantee and potentially destabilizing the entire system. This cascade failure mode is not
addressed by safety filters acting at the outer loop alone, nor by saturating the outer-loop-issued commands. Mitigating this risk requires
a safety filter at the inner-loop level, backed by a formal invariance certificate that is valid
for the full nonlinear attitude dynamics and robust to external disturbances.

\IEEEpubidadjcol Safety filters have emerged as an effective supervisory framework for enforcing safety in
constrained autonomous systems~\cite{hsu_safety_2023, wabersich_data-driven_2023}. They operate in
parallel with nominal controllers, monitoring system states and modifying control inputs only when
necessary to prevent constraint violations. The dominant approach to safety filter design relies on
Control Barrier Functions (CBFs)~\cite{ames_control_2019}, which impose state-dependent constraints
on control inputs to enforce forward invariance of safe sets. Hamilton-Jacobi reachability methods
provide an alternative by solving partial differential equations to compute reach-avoid
sets~\cite{bansal_hamilton-jacobi_2017}, though they suffer from the curse of dimensionality. For
robust safety under bounded disturbances, CBF-based methods have been extended via
Hamilton-Jacobi-Isaacs formulations~\cite{choi_robust_2021}, input-to-state safe
CBFs~\cite{kolathaya_input--state_2019, alan_control_2022}, and robust CBF
constructions~\cite{jankovic_robust_2018}. Recent work has applied these tools directly to attitude
safety: run-time assurance frameworks enforce multiple simultaneous attitude constraints for
spacecraft maneuvering via active set invariance filters~\cite{mcquinn2024runtimeassurancesimultaneous},
and backup CBFs have been used to aggregate certified safe regions for spacecraft attitude control
and station keeping~\cite{ong2026safespaceaggregatingsafesets}.
Despite their flexibility, CBF-based approaches share a fundamental limitation: the barrier function must be constructed manually, and no systematic procedure
exists for quaternion attitude dynamics. Furthermore, guaranteeing feasibility of the CBF quadratic
program under input constraints requires an explicitly characterized control invariant set, which is
itself difficult to compute for nonlinear systems~\cite{chen_backup_2021}.

Backup CBFs address the invariant set construction challenge by implicitly defining a control
invariant set through forward simulation of a fixed backup controller~\cite{chen_backup_2021,
singletary_onboard_2022}. This approach has been demonstrated on high-speed racing drones for
geofencing and extended to systems under bounded disturbances via tube-based
formulations~\cite{wijk_disturbance-robust_2024}. However, the implicit safe set is evaluated
online by numerical integration, which carries a computational cost that grows with the prediction
horizon and state dimension. More critically, the certified region depends on the chosen backup
policy and horizon length, with no closed-form characterization and no systematic procedure for
jointly optimizing the backup controller and the invariant set subject to state and input
constraints.

In contrast, set-theoretic methods directly compute robust controlled invariant (RCI) sets,
geometric regions within which safety can be maintained by appropriate control actions under all
admissible disturbances~\cite{blanchini_set-theoretic_2015, lygeros2004reachability}. For linear
systems, ellipsoidal RCI sets can be efficiently computed using Linear Matrix Inequality (LMI)
optimization~\cite{boyd_linear_1994}. Early work by Usoro et al.~\cite{usoro1982ellipsoidal}
established fundamental results on computing ellipsoidal invariant sets that maximize disturbance
tolerance while satisfying state and input constraints, later systematized by the invariant
ellipsoid technique of Khlebnikov et al.~\cite{khlebnikov_optimization_2011}. Invariant sets have
since served as terminal safe regions in predictive safety
filters~\cite{wabersich_predictive_2021, li_learning_2023} and as motion-planning tools for
constrained attitude maneuvering~\cite{danielson2022invariant}.
However, none of these works synthesize an ellipsoidal RCI set with formal safety guarantees for
nonlinear quaternion attitude dynamics. The challenge is that the quaternion kinematic nonlinearity breaks the linearity required for standard LMI synthesis, and existing approaches either linearize around an equilibrium, losing the formal guarantee for large rotation angles, or treat the nonlinearity as a norm-bounded uncertainty, introducing conservatism that grows with the rotation envelope.

This paper addresses these challenges with three main contributions. First, we formulate safety
filter synthesis for linear systems as a unified LMI optimization that simultaneously computes a
maximal ellipsoidal RCI set and its associated state-feedback backup control law, subject to state
and input constraints, recovering and unifying results
from~\cite{blanchini_set-theoretic_2015, khlebnikov_optimization_2011} within a safety filter
synthesis context. Second, and as the primary theoretical contribution, we extend this framework to
nonlinear quaternion attitude dynamics without approximation. Exact, closed-form sector bounds on
the quaternion kinematic nonlinearity are derived analytically from the geometric orthogonality of
the radial and tangential components of the kinematic map (Lemma~2). The resulting sector bound
depends only on the maximum admissible rotation angle, requires no numerical Lipschitz estimation,
and is embedded directly into the LMI synthesis via the S-procedure~\cite{boyd_linear_1994,
polik2007survey}, yielding a single convex semidefinite program whose solution is simultaneously
the largest certifiable safe region and its associated backup control law. Third,
we demonstrate the framework on the inner loop of a hierarchical quadrotor controller, providing a
formal guarantee that the attitude state remains within its certified safe region regardless of the
commands generated by the outer position loop, directly addressing the cascade failure mode
described above. An illustration of the proposed framework is shown in
Fig.~\ref{fig:graphical_abstract}.

The remainder of this paper is organized as follows. Section~\ref{sec:preliminaries} establishes the
mathematical foundations, including the system model, viable sets, and ellipsoidal robust
controlled invariant sets. Section~\ref{sec:ellipsoidal_safety_conditions} develops the LMI formulation for linear systems
and the associated constraint satisfaction conditions. Section~\ref{sec:robust_safety_filter_synthesis} presents the
unified optimization problem and the smooth filtering strategy. Section~\ref{sec:quaternion_safety_filter}
extends the framework to nonlinear quaternion attitude dynamics via the exact sector bound and
S-procedure embedding. Section~\ref{sec:numerical_results} presents numerical simulations on a quadrotor
system, and Section~\ref{sec:conclusion} concludes with directions for future work.

\section{Preliminaries and Problem Formulation}\label{sec:preliminaries}

This section establishes the mathematical foundations for robust safety filter synthesis. We begin by introducing the system model and safety constraints, then develop the key concepts of viable sets and RCI sets that form the basis of our approach.

\subsection{System Model and Safety Constraints}

Consider a continuous-time dynamical system with control inputs and disturbances:
\begin{equation}
    \label{eq:system_dynamics}
    \dot{\bm{x}}(t) = f(\bm{x}(t), \bm{u}(t), \bm{d}(t)), \quad \bm{x}(0) = \bm{x}_0
\end{equation}
where $\bm{x}(t) \in \mathbb{R}^n$ is the state vector, $\bm{u}(t) \in \mathcal{U} \subseteq \mathbb{R}^m$ is the bounded control input,
and $\bm{d}(t) \in \mathcal{D} \subseteq \mathbb{R}^p$ is the bounded disturbance vector.
The control input set $\mathcal{U}$ and disturbance set $\mathcal{D}$ are assumed to be compact.
The function $f: \mathbb{R}^n \times \mathbb{R}^m \times \mathbb{R}^p \to \mathbb{R}^n$ is assumed to be locally Lipschitz continuous
to ensure existence and uniqueness of solutions.

The system is subject to safety constraints that define a safe operating region:
\begin{equation}
\label{eq:safety_constraints}
g_j(\bm{x}) \leq 0, \quad j = 1, \ldots, \ell
\end{equation}
where $g_j: \mathbb{R}^n \to \mathbb{R}$ are continuously differentiable constraint functions.

\begin{definition}[Safe Set]
\label{def:safe_set}
The safe set $\mathcal{S} \subseteq \mathbb{R}^n$ is defined as the region of the state space where all safety constraints are satisfied:
\begin{equation}
\mathcal{S} = \{\bm{x} \in \mathbb{R}^n : g_j(\bm{x}) \leq 0, \, j = 1, \ldots, \ell\}
\end{equation}
The safe set $\mathcal{S}$ is assumed to be non-empty and closed.
\end{definition}

\subsection{Viable Sets and Controlled Invariance}

To develop safety filters that provide formal guarantees, we require concepts from viability theory and controlled invariance. These mathematical tools allow us to characterize regions of the state space from which safety can be maintained despite uncertainties.

\begin{definition}[Viable Set \cite{lygeros2004reachability}]
\label{def:viable_set}
A set $\mathcal{V} \subseteq \mathcal{S}$ is a viable set for the system \eqref{eq:system_dynamics} if for every initial state $\bm{x}_0 \in \mathcal{V}$, there exists an admissible control law $\bm{u}(t) \in \mathcal{U}$ such that the state trajectory $\bm{x}(t)$ remains in $\mathcal{S}$ for all $t \geq 0$, regardless of the disturbance realization $\bm{d}(t) \in \mathcal{D}$.
The maximal viable set (or viability kernel), denoted as $\mathcal{V}^*$, is the largest set satisfying this property.
\end{definition}

While the maximal viable set provides the theoretical limit of states from which safety can be guaranteed, it is often difficult to compute exactly for complex systems. A more tractable approach is to find specific controlled invariant sets with known control laws.

\begin{definition}[Safe Robust Controlled Invariant Set]
\label{def:robust_controlled_invariant}
A set $\mathcal{C} \subseteq \mathcal{S}$ is a safe RCI set 
for the system \eqref{eq:system_dynamics} under a given control law $\bm{u} = \pi(\bm{x})$ with $\pi: \mathbb{R}^n \to \mathcal{U}$ if for every initial state $\bm{x}_0 \in \mathcal{C}$, the closed-loop state trajectory satisfies
\begin{equation}
\bm{x}(t) \in \mathcal{C}, \quad \forall t \geq 0, \quad \forall \bm{d}(t) \in \mathcal{D}.
\end{equation}
\end{definition}

\begin{remark}
Every safe RCI set is a subset of the maximal viable set, i.e., $\mathcal{C} \subseteq \mathcal{V}^*$.
\end{remark}

\subsection{Nagumo's Theorem for Controlled Systems}
Nagumo's theorem provides necessary and sufficient conditions for a set to be invariant under a given dynamical system. We extend this classical result to controlled systems with disturbances, which is essential for verifying the invariance of candidate sets in safety filter design.
To state Nagumo's theorem, we first need to define the tangent cone to a set at a given point.

\begin{definition}[Bouligand's Tangent Cone \cite{blanchini_set-theoretic_2015}]
\label{def:tangent_cone}
For a set $\mathcal{C} \subseteq \mathbb{R}^n$ and a point $\bm{x} \in \mathcal{C}$, the Bouligand tangent cone to $\mathcal{C}$ at $\bm{x}$, denoted $T_{\mathcal{C}}(\bm{x})$, is defined as
\begin{equation}
T_{\mathcal{C}}(\bm{x}) = \left\{\bm{v} \in \mathbb{R}^n : \liminf_{h \to 0^+} \frac{d(\bm{x} + h\bm{v}, \mathcal{C})}{h} = 0\right\}
\end{equation}
where $d(\bm{y}, \mathcal{C}) = \inf_{\bm{z} \in \mathcal{C}} \|\bm{y} - \bm{z}\|$ is the distance from point $\bm{y}$ to the set $\mathcal{C}$.
\end{definition}

\begin{remark}
The tangent cone at a boundary point $\bm{x}$ of a set $\mathcal{C}$ can be intuitively understood as the set of all possible directions in which one can "move" from $\bm{x}$ while remaining within or on the boundary of $\mathcal{C}$.
\end{remark}

For sets defined by smooth functions, the tangent cone has a particularly simple characterization:

\begin{proposition}
    \label{prop:tangent_cone_smooth}
    If the set $\mathcal{C}$ is defined as the sublevel set of a smooth function $h: \mathbb{R}^n \to \mathbb{R}$, i.e., $\mathcal{C} = \{\bm{x} \in \mathbb{R}^n : h(\bm{x}) \leq 0\}$,
    then the tangent cone at a boundary point $\bm{x} \in \partial \mathcal{C}$ (where $h(\bm{x}) = 0$ and $\nabla h(\bm{x}) \neq 0$) can be characterized as
    \begin{equation}
        T_{\mathcal{C}}(\bm{x}) = \{\bm{v} \in \mathbb{R}^n : \nabla h(\bm{x})^T \bm{v} \leq 0\}
    \end{equation}
\end{proposition}
\begin{proof}
At a regular boundary point $\bm{x} \in \partial \mathcal{C}$ where $h(\bm{x}) = 0$ and $\nabla h(\bm{x}) \neq 0$, 
the gradient $\nabla h(\bm{x})$ is the outward normal to the boundary. 
By Taylor expansion, $h(\bm{x} + h\bm{v}) = h\nabla h(\bm{x})^T \bm{v} + o(h)$.
A direction $\bm{v}$ belongs to the tangent cone if and only if $\bm{x} + h\bm{v}$ remains in or 
approaches $\mathcal{C}$ as $h \to 0^+$, which occurs precisely when $\nabla h(\bm{x})^T \bm{v} \leq 0$.
See \cite{aubin2009set}, Chapter 4, for details.
\end{proof}
Using the tangent cone, we can now state Nagumo's theorem for controlled systems with disturbances.

\begin{theorem}
    \label{thm:nagumo_controlled}
    (Nagumo's Theorem for Controlled Systems with Disturbance \cite{blanchini_set-theoretic_2015})
    Consider the controlled system with disturbances \eqref{eq:system_dynamics}.
    Let $\mathcal{C} \subseteq \mathbb{R}^n$ be a closed set with boundary $\partial \mathcal{C}$ and tangent cone $T_{\mathcal{C}}(\bm{x})$ at $\bm{x} \in \partial \mathcal{C}$.
    The set $\mathcal{C}$ is a RCI set under a given control law $\bm{u} = \pi(\bm{x})$ with $\pi: \mathbb{R}^n \to \mathcal{U}$ if and only if for every $\bm{x} \in \partial \mathcal{C}$:
    \begin{equation}
        f(\bm{x}, \pi(\bm{x}), \bm{d}) \in T_{\mathcal{C}}(\bm{x}), \quad \forall \bm{d} \in \mathcal{D}
    \end{equation}
\end{theorem}
\begin{proof}
We reformulate the controlled system with disturbances as a differential inclusion. 
For a fixed control law $\pi: \mathbb{R}^n \to \mathcal{U}$, the closed-loop system 
can be written as $\dot{\bm{x}} \in F(\bm{x}) := \{f(\bm{x}, \pi(\bm{x}), \bm{d}) : \bm{d} \in \mathcal{D}\}$.
By Nagumo's theorem for differential inclusions \cite{aubin2011viability}, 
a closed set $\mathcal{C}$ is invariant if and only if 
$F(\bm{x}) \subseteq T_{\mathcal{C}}(\bm{x})$ for all $\bm{x} \in \partial \mathcal{C}$,
which is equivalent to the stated condition.
\end{proof}

\subsection{Ellipsoidal Sets}

In this work, we focus on ellipsoidal sets due to their computational tractability and rich geometric properties that allow for efficient LMI-based synthesis procedures.

\begin{definition}[Ellipsoidal Set]
    \label{def:ellipsoidal_set}
        An ellipsoidal set $\mathcal{E} \subseteq \mathbb{R}^n$ centered at $\bm{c} \in \mathbb{R}^n$ with shape matrix $\mathbf{P} \in \mathbb{R}^{n \times n}$ is defined as
        \begin{equation}
            \mathcal{E}(\bm{c}, \mathbf{P}) = \{\bm{x} \in \mathbb{R}^n : (\bm{x} - \bm{c})^T \mathbf{P} (\bm{x} - \bm{c}) \leq 1\}
        \end{equation}
        where $\mathbf{P} \succ 0$ is a symmetric positive definite matrix. When the ellipsoid is centered at the origin ($\bm{c} = \bm{0}$), we use the simplified notation $\mathcal{E}(\mathbf{P}) = \{\bm{x} \in \mathbb{R}^n : \bm{x}^T \mathbf{P} \bm{x} \leq 1\}$.
\end{definition}

\subsection{Safety Filter Definition and Problem Statement}

With the mathematical foundations established, we now formally define the robust safety filter and state the synthesis problem addressed in this work.

\begin{definition}[Robust Safety Filter \cite{wabersich_data-driven_2023}]
\label{def:safety_filter}
Given a safe set $\mathcal{S} \subseteq \mathbb{R}^n$, a robust safety filter 
$\pi_s: \mathbb{R}^n \times \mathbb{R}^m \to \mathbb{R}^m$ modifies a nominal (desired) control input 
$\bm{u}_{\text{nom}}(t)$ to produce a filtered control input
\begin{equation}
 \bm{u}_s(t) = \pi_s(\bm{x}(t), \bm{u}_{\text{nom}}(t))
\end{equation}
that ensures the system trajectory satisfies $\bm{x}(t) \in \mathcal{S}$ for all $t \geq 0$ 
and all disturbances $\bm{d}(t) \in \mathcal{D}$, while minimally modifying the nominal input signal.
\end{definition}

\begin{remark}
In practice, a robust safety filter is often designed with respect to a RCI 
set $\mathcal{C} \subseteq \mathcal{S}$, ensuring that trajectories starting in $\mathcal{C}$ 
remain within both $\mathcal{C}$ and $\mathcal{S}$. This approach provides computational tractability while maintaining safety guarantees.
\end{remark}

\begin{problem}
\label{prob:safety_filter_synthesis}
Given the controlled system with disturbances \eqref{eq:system_dynamics}
and a safe set $\mathcal{S}$ defined by safety constraints \eqref{eq:safety_constraints},
synthesize a robust safety filter $\pi_s: \mathbb{R}^n \times \mathbb{R}^m \to \mathbb{R}^m$ that provides a constructive design procedure that is computationally tractable for real-time implementation.
The filter should ensure that the closed-loop system satisfies $\bm{x}(t) \in \mathcal{S}$ for all $t \geq 0$ and all disturbances $\bm{d}(t) \in \mathcal{D}$.
\end{problem}

\section{Ellipsoidal Set-Theoretic Safety Conditions}
\label{sec:ellipsoidal_safety_conditions}

This section develops the mathematical framework for robust safety filter synthesis using ellipsoidal sets. We begin by specializing Nagumo's theorem to ellipsoidal sets, then formulate the design problem as a tractable LMI optimization for linear systems.

\subsection{Nagumo's Theorem for Ellipsoidal Sets}

For ellipsoidal sets, the tangent cone at boundary points has a particularly elegant characterization that enables efficient computational methods.

\begin{corollary}
    \label{cor:tangent_cone_ellipsoid}
    For an ellipsoidal set $\mathcal{E}(\bm{c}, \mathbf{P})$, the tangent cone at a boundary point $\bm{x} \in \partial \mathcal{E}$ is characterized as
    \begin{equation}
        T_{\mathcal{E}}(\bm{x}) = \{\bm{v} \in \mathbb{R}^n : (\bm{x} - \bm{c})^T \mathbf{P} \bm{v} \leq 0\}
    \end{equation}
\end{corollary}
\begin{proof}
    The proof follows directly from Proposition \ref{prop:tangent_cone_smooth} by choosing $h(\bm{x}) = (\bm{x} - \bm{c})^T \mathbf{P} (\bm{x} - \bm{c}) - 1$.
\end{proof}

This geometric characterization leads to a simple invariance condition for ellipsoidal sets under controlled dynamics.

\begin{corollary}
    \label{cor:nagumo_ellipsoid}
    (Nagumo's Theorem for Ellipsoidal Sets)
    An ellipsoidal set $\mathcal{E}(\bm{c}, \mathbf{P})$ is a RCI set for the system \eqref{eq:system_dynamics} under control law $\bm{u} = \pi(\bm{x})$ if and only if for every $\bm{x} \in \partial \mathcal{E}$:
    \begin{equation}
        (\bm{x} - \bm{c})^T \mathbf{P} f(\bm{x}, \pi(\bm{x}), \bm{d}) \leq 0, \quad \forall \bm{d} \in \mathcal{D}
    \end{equation}
\end{corollary}
\begin{proof}
    This follows directly from applying Theorem \ref{thm:nagumo_controlled} with the tangent cone characterization from Corollary \ref{cor:tangent_cone_ellipsoid}.
\end{proof}

\subsection{LMI Formulation for Linear Systems}

We now specialize our approach to linear time-invariant systems, where the ellipsoidal RCI set synthesis becomes a convex optimization problem.

Consider the linear system with bounded disturbances:
\begin{equation}
    \label{eq:linear_system}
    \dot{\bm{x}}(t) = \mathbf{A}\bm{x}(t) + \mathbf{B}\bm{u}(t) + \mathbf{E}\bm{d}(t)
\end{equation}
where $\mathbf{A} \in \mathbb{R}^{n \times n}$, $\mathbf{B} \in \mathbb{R}^{n \times m}$, and $\mathbf{E} \in \mathbb{R}^{n \times p}$ are system matrices, and the disturbance satisfies $\|\bm{d}(t)\|_2 \leq 1$ for all $t \geq 0$.

For linear systems with ellipsoidal sets centered at the origin, Nagumo's condition simplifies considerably.

\begin{corollary}
    \label{cor:nagumo_linear_ellipsoid}
    (Invariance Condition for Linear Systems)
    An ellipsoidal set $\mathcal{E}(\mathbf{P})$ centered at the origin is RCI for the linear system \eqref{eq:linear_system} if there exists a control law $\bm{u} = \pi(\bm{x})$ such that
    \begin{align}
        \bm{x}^T \mathbf{P} (\mathbf{A}\bm{x} + \mathbf{B}\pi(\bm{x}) + \mathbf{E}\bm{d}) \leq 0
    \end{align}
    for all $\bm{d}$ with $\|\bm{d}\|_2 \leq 1$ and all $\bm{x} \in \partial\mathcal{E}(\mathbf{P})$.
\end{corollary}
\begin{proof}
    Direct substitution of the linear system dynamics $f(\bm{x}, \bm{u}, \bm{d}) = \mathbf{A}\bm{x} + \mathbf{B}\bm{u} + \mathbf{E}\bm{d}$ 
    into Corollary \ref{cor:nagumo_ellipsoid}.
\end{proof}

To convert this condition into a tractable optimization problem, we employ the S-procedure lemma:

\begin{lemma}[S-procedure \cite{boyd_linear_1994}]
    \label{lem:s_procedure}
    Consider quadratic forms $f_i(\bm{x}) = \bm{x}^T \mathbf{A}_i \bm{x}$ for $i = 0, 1, \ldots, m$, with $\mathbf{A}_i = \mathbf{A}_i^T$. If there exist $\tau_i \geq 0$ for $i = 1, \ldots, m$ such that
    \begin{equation}
        \mathbf{A}_0 \preceq \sum_{i=1}^{m} \tau_i \mathbf{A}_i, \qquad \alpha_0 \geq \sum_{i=1}^{m} \tau_i \alpha_i,
    \end{equation}
    then the constraints $f_i(\bm{x}) \leq \alpha_i$ for $i = 1, \ldots, m$ imply $f_0(\bm{x}) \leq \alpha_0$.
\end{lemma}

\begin{proof}
    The proof can be found in \cite{boyd_linear_1994}.
\end{proof}

\begin{remark}
When one or more inequalities become equalities, the nonnegativity requirement on the corresponding $\tau_i$ is omitted \cite{khlebnikov_optimization_2011}.
\end{remark}

Using the S-procedure, we can now formulate the RCI set synthesis as an LMI optimization problem.

\begin{theorem}
    \label{thm:lmi_rfcis}
    (LMI Formulation for RCI Set Synthesis)
    The ellipsoidal set $\mathcal{E}(\mathbf{P})$ is RCI for the linear system \eqref{eq:linear_system} with state feedback control law $\bm{u} = \mathbf{K}\bm{x}$ if there exist matrices $\mathbf{Q} = \mathbf{P}^{-1} \succ 0$ and $\mathbf{Y}$ such that
    \begin{equation}
        \label{eq:main_lmi}
        \begin{bmatrix}
            \mathbf{A}\mathbf{Q} + \mathbf{Q}\mathbf{A}^T + \mathbf{B}\mathbf{Y} + \mathbf{Y}^T\mathbf{B}^T + \lambda\mathbf{Q} & \mathbf{E} \\
            \mathbf{E}^T & -\lambda \mathbf{I}
        \end{bmatrix} \preceq 0
    \end{equation}
    for some scalar $\lambda > 0$. The feedback gain is recovered as $\mathbf{K} = \mathbf{Y}\mathbf{Q}^{-1}$.
\end{theorem}

\begin{proof}
    Substituting $\bm{u} = \mathbf{K}\bm{x}$ into Corollary~\ref{cor:nagumo_linear_ellipsoid} and writing the invariance condition as a quadratic form in the augmented vector $\bm{z} = [\bm{x}^T, \bm{d}^T]^T$, the requirement is
    \begin{equation}
        \bm{z}^T \begin{bmatrix}
            \mathbf{P}(\mathbf{A}+\mathbf{B}\mathbf{K})+(\mathbf{A}+\mathbf{B}\mathbf{K})^T\mathbf{P} & \mathbf{P}\mathbf{E} \\
            \mathbf{E}^T\mathbf{P} & \mathbf{0}
        \end{bmatrix} \bm{z} \leq 0
    \end{equation}
    for all $\bm{z}$ satisfying $\bm{x}^T\mathbf{P}\bm{x} = 1$ and $\bm{d}^T\bm{d} \leq 1$. Applying the S-procedure (Lemma~\ref{lem:s_procedure}) with an unrestricted multiplier $\tau_1$ for the equality constraint and $\tau_2 \geq 0$ for the disturbance ball, and choosing $\tau_1 = -\lambda$, $\tau_2 = \lambda$ for some $\lambda > 0$, yields
    \begin{equation}
        \begin{bmatrix}
            \mathbf{P}(\mathbf{A}+\mathbf{B}\mathbf{K})+(\mathbf{A}+\mathbf{B}\mathbf{K})^T\mathbf{P} + \lambda\mathbf{P} & \mathbf{P}\mathbf{E} \\
            \mathbf{E}^T\mathbf{P} & -\lambda\mathbf{I}
        \end{bmatrix} \preceq 0.
    \end{equation}
    Applying the congruence transformation $\mathrm{diag}(\mathbf{Q}, \mathbf{I})$ with $\mathbf{Q} = \mathbf{P}^{-1}$ and substituting $\mathbf{Y} = \mathbf{K}\mathbf{Q}$ gives the LMI~\eqref{eq:main_lmi}.
\end{proof}

\subsection{Constraint Satisfaction}

For the linear system \eqref{eq:linear_system}, we consider safety constraints on the system outputs of the form:
\begin{equation}
 \mathbf{C}\bm{x} \leq \bm{y}_{\max}
\end{equation}
where $\mathbf{C} \in \mathbb{R}^{q \times n}$ is the output matrix and $\bm{y}_{\max} \in \mathbb{R}^{q}$ defines the safe operating region.
Additionally, the control inputs are subject to box constraints:
\begin{equation}
    |u_i| \leq u_{i,\max}, \quad i = 1, \ldots, m
\end{equation}
where $u_{i,\max}$ is the maximum allowable magnitude for the $i$-th control input. 
We now show how these constraints can be incorporated into the LMI framework to ensure that trajectories starting within the ellipsoidal RCI set remain safe and feasible.

\begin{proposition}
    \label{prop:input_constraint}
    (Input Constraint Satisfaction)
    The state feedback control law $\bm{u} = \mathbf{K}\bm{x}$ satisfies individual input constraints $|u_i| \leq u_{i,\max}$ for all $\bm{x} \in \mathcal{E}(\mathbf{P})$ if
    \begin{equation}
        \begin{bmatrix}
            \mathbf{Q} & \mathbf{y}_i^T \\
            \mathbf{y}_i & u_{i,\max}^2
        \end{bmatrix} \succeq 0, \quad i = 1, \ldots, m
    \end{equation}
    where $\mathbf{y}_i^T$ is the $i$-th row of $\mathbf{Y}$.
\end{proposition}

\begin{proof}
    For $\bm{x} \in \mathcal{E}(\mathbf{P})$, we have $\bm{x}^T\mathbf{P}\bm{x} \leq 1$. The constraint $|u_i| = |\mathbf{k}_i^T\bm{x}| \leq u_{i,\max}$ is equivalent to $(\mathbf{k}_i^T\bm{x})^2 \leq u_{i,\max}^2$, where $\mathbf{k}_i^T$ is the $i$-th row of $\mathbf{K}$.
    We formulate this using quadratic forms with $f_0(\bm{x}) = (\mathbf{k}_i^T\bm{x})^2 = \bm{x}^T(\mathbf{k}_i\mathbf{k}_i^T)\bm{x}$ and $f_1(\bm{x}) = \bm{x}^T\mathbf{P}\bm{x}$.
    Applying the S-procedure (Lemma \ref{lem:s_procedure}) with $\alpha_0 = u_{i,\max}^2$ and $\alpha_1 = 1$, there exists a multiplier $\tau_i \geq 0$ such that
    \begin{equation}
        \mathbf{k}_i\mathbf{k}_i^T \preceq \tau_i \mathbf{P} \quad \text{and} \quad \tau_i \leq u_{i,\max}^2
    \end{equation}
    Substituting $\mathbf{P} = \mathbf{Q}^{-1}$ and pre- and post-multiplying by $\mathbf{Q}$ gives
    \begin{equation}
        \mathbf{Q}\mathbf{k}_i\mathbf{k}_i^T\mathbf{Q} \preceq \tau_i \mathbf{Q}
    \end{equation}
    From the definition $\mathbf{Y} = \mathbf{K}\mathbf{Q}$, we have $\mathbf{k}_i^T = \mathbf{y}_i^T\mathbf{Q}^{-1}$, which gives $\mathbf{k}_i = \mathbf{Q}^{-1}\mathbf{y}_i$. Substituting this:
    \begin{equation}
        \mathbf{Q}(\mathbf{Q}^{-1}\mathbf{y}_i)(\mathbf{Q}^{-1}\mathbf{y}_i)^T\mathbf{Q} = \mathbf{y}_i\mathbf{y}_i^T \preceq \tau_i \mathbf{Q}
    \end{equation}
    Dividing by $\tau_i$ and rearranging gives
    \begin{equation}
        \frac{1}{\tau_i} \mathbf{y}_i\mathbf{y}_i^T \preceq \mathbf{Q}
    \end{equation}
    By the Schur complement lemma, and setting $\tau_i = u_{i,\max}^2$, this is equivalent to
    \begin{equation}
        \begin{bmatrix}
            \mathbf{Q} & \mathbf{y}_i^T \\
            \mathbf{y}_i & u_{i,\max}^2
        \end{bmatrix} \succeq 0
    \end{equation}
    This completes the proof.
\end{proof}

\begin{proposition}
    \label{prop:output_constraint}
    (Output Constraint Satisfaction)
    For the system output $\bm{y} = \mathbf{C}\bm{x}$, individual output constraints $|y_j| \leq y_{j,\max}$ are satisfied for all $\bm{x} \in \mathcal{E}(\mathbf{P})$ if
    \begin{equation}
        \begin{bmatrix}
            \mathbf{Q} & \mathbf{Q}\mathbf{c}_j \\
            \mathbf{c}_j^T\mathbf{Q} & y_{j,\max}^2
        \end{bmatrix} \succeq 0, \quad j = 1, \ldots, q
    \end{equation}
    where $\mathbf{c}_j^T$ is the $j$-th row of $\mathbf{C}$.
\end{proposition}
\begin{proof}
    For $\bm{x} \in \mathcal{E}(\mathbf{P})$, we have $\bm{x}^T\mathbf{P}\bm{x} \leq 1$. The maximum of $|y_j| = |\mathbf{c}_j^T\bm{x}|$ over $\mathcal{E}(\mathbf{P})$ equals $\sqrt{\mathbf{c}_j^T\mathbf{P}^{-1}\mathbf{c}_j} = \sqrt{\mathbf{c}_j^T\mathbf{Q}\mathbf{c}_j}$, so the constraint requires $\mathbf{c}_j^T\mathbf{Q}\mathbf{c}_j \leq y_{j,\max}^2$. By the Schur complement lemma applied to the $(2,2)$ block, $\begin{bmatrix}\mathbf{Q} & \mathbf{Q}\mathbf{c}_j \\ \mathbf{c}_j^T\mathbf{Q} & y_{j,\max}^2\end{bmatrix} \succeq 0$ is equivalent to $y_{j,\max}^2 - \mathbf{c}_j^T\mathbf{Q}\mathbf{Q}^{-1}\mathbf{Q}\mathbf{c}_j = y_{j,\max}^2 - \mathbf{c}_j^T\mathbf{Q}\mathbf{c}_j \geq 0$.
\end{proof}

\section{Robust Safety Filter Synthesis}
\label{sec:robust_safety_filter_synthesis}
To synthesize a robust safety filter, we combine the invariance conditions with input and output constraints into a unified optimization framework. 
The key insight is to simultaneously compute the largest ellipsoidal RCI set while ensuring all physical and safety constraints are satisfied. Maximizing the volume of the ellipsoidal RCI set is crucial for practical implementation as it expands the safe operating region, reducing the conservatism of the safety filter. A larger invariant set allows the nominal controller to operate freely over a wider state space before intervention becomes necessary.

The volume of an ellipsoid $\mathcal{E}(\mathbf{P})$ is proportional to $(\det \mathbf{P})^{-1/2} = (\det \mathbf{Q})^{1/2}$ where $\mathbf{Q} = \mathbf{P}^{-1}$. Since $\log\det\mathbf{Q}$ is concave in $\mathbf{Q} \succ 0$, maximizing it yields a convex program that is directly supported by modern SDP solvers (e.g., CVX, YALMIP, MOSEK).

\subsection{Unified LMI Optimization Formulation}
By combining the invariance condition from Theorem \ref{thm:lmi_rfcis} with the constraint satisfaction conditions from Propositions \ref{prop:input_constraint} and \ref{prop:output_constraint}, we obtain the following unified formulation:

\begin{problem}
    \label{prob:unified_lmi_optimization}
    (Unified LMI Formulation for Robust Safety Filter Synthesis)
    The robust safety filter synthesis problem for the linear system \eqref{eq:linear_system} with input constraints $|u_i| \leq u_{i,\max}$ and output constraints $|y_j| \leq y_{j,\max}$ can be formulated as the following LMI optimization problem:
    
    \begin{align}
        \max_{\mathbf{Q}, \mathbf{Y}} \quad & \log\det(\mathbf{Q}) \\
        \text{s.t.} \quad & \begin{bmatrix}
            \mathbf{A}\mathbf{Q} + \mathbf{Q}\mathbf{A}^T + \mathbf{B}\mathbf{Y} + \mathbf{Y}^T\mathbf{B}^T + \lambda\mathbf{Q} & \mathbf{E} \\
            \mathbf{E}^T & -\lambda \mathbf{I}
        \end{bmatrix} \preceq 0 \\
        & \begin{bmatrix}
            \mathbf{Q} & \mathbf{y}_i^T \\
            \mathbf{y}_i & u_{i,\max}^2
        \end{bmatrix} \succeq 0, \quad i = 1, \ldots, m \\
        & \begin{bmatrix}
            \mathbf{Q} & \mathbf{Q}\mathbf{c}_j \\
            \mathbf{c}_j^T\mathbf{Q} & y_{j,\max}^2
        \end{bmatrix} \succeq 0, \quad j = 1, \ldots, q \\
        & \mathbf{Q} \succ 0
    \end{align}
    where $\lambda > 0$ is a fixed parameter, and the decision variables are $\mathbf{Q} = \mathbf{P}^{-1}$ and $\mathbf{Y} = \mathbf{K}\mathbf{Q}$. The optimal ellipsoidal RCI set is $\mathcal{E}(\mathbf{P}^*)$ with $\mathbf{P}^* = (\mathbf{Q}^*)^{-1}$, and the associated control law is $\bm{u} = \mathbf{K}^*\bm{x}$ with $\mathbf{K}^* = \mathbf{Y}^*(\mathbf{Q}^*)^{-1}$.
\end{problem}

\begin{remark}
    The optimization problem in Problem \ref{prob:unified_lmi_optimization} is a convex semidefinite program (SDP) that can be efficiently solved using standard numerical solvers such as MOSEK, SeDuMi, or SDPT3. The convexity ensures global optimality of the solution, providing a reliable and computationally tractable method for robust safety filter synthesis.
\end{remark}

\subsection{Filtering Strategy}
Once the RCI set $\mathcal{E}(\mathbf{P}^*)$ and the associated safe control law $\bm{u}_{\text{b}} = \mathbf{K}^*\bm{x}$ are determined, a safety filter can be constructed to minimally intervene with a nominal performance controller $\bm{u}_{\text{nom}}$. The goal is to use $\bm{u}_{\text{nom}}$ when the system is safely inside the RCI set and smoothly switch to $\bm{u}_{\text{b}}$ as the state approaches the boundary.

This is achieved using a mixing function $\alpha(\bm{x})$ that depends on the safety metric $h(\bm{x}) = \bm{x}^T\mathbf{P}^*\bm{x}$. The filtered control input $\bm{u}_s$ is given by:
\begin{equation}
    \label{eq:safety_filter_law}
    \bm{u}_s(\bm{x}, \bm{u}_{\text{nom}},\bm{u}_{\text{b}}) = (1 - \alpha(\bm{x}))\bm{u}_{\text{nom}} + \alpha(\bm{x})\bm{u}_{\text{b}}
\end{equation}
The mixing function $\alpha(\bm{x})$ is designed to be 0 deep inside the set and 1 near the boundary. A common choice is a ramp function:
\begin{equation}
    \label{eq:mixing_function}
    \alpha(\bm{x}) = 
    \begin{cases}
        0 & \text{if } h(\bm{x}) \leq h_{\min} \\
        \frac{h(\bm{x}) - h_{\min}}{h_{\max} - h_{\min}} & \text{if } h_{\min} < h(\bm{x}) < h_{\max} \\
        1 & \text{if } h(\bm{x}) \geq h_{\max}
    \end{cases}
\end{equation}
where $0 < h_{\min} < h_{\max} \leq 1$ are tunable parameters defining a transition region.

When the state $\bm{x}$ is on the boundary of the RCI set, i.e., $\bm{x}^T\mathbf{P}^*\bm{x} = 1$, we have $h(\bm{x}) = 1$. Since $h_{\max} \leq 1$, this implies $\alpha(\bm{x}) = 1$. Consequently, the safety filter outputs $\bm{u}_s = \bm{u}_{\text{b}} = \mathbf{K}^*\bm{x}$. As established by Theorem \ref{thm:lmi_rfcis}, this control law ensures that the system's velocity vector points inwards or is tangent to the ellipsoid, satisfying the Nagumo condition and thus guaranteeing the forward invariance of the set $\mathcal{E}(\mathbf{P}^*)$.

\begin{remark}
    The mixing function $\alpha(\bm{x})$ can be designed to be smooth (e.g., using a sigmoid function) to ensure continuity and differentiability of the filtered control law $\bm{u}_s$. 
\end{remark}

\begin{remark}
    If both $\bm{u}_{\text{nom}}$ and $\bm{u}_{\text{b}}$ satisfy input constraints, then their convex combination $\bm{u}_s$ also satisfies the constraints.
\end{remark}

\begin{remark}
    \label{rem:nonlinear_extension}
    (Linearization-Based Extension)
    For nonlinear systems $\dot{\bm{x}} = f(\bm{x}, \bm{u}) + \mathbf{E}_d\bm{d}$ that admit a decomposition $f(\bm{x},\bm{u}) = \mathbf{A}\bm{x} + \mathbf{B}\bm{u} + \mathbf{B}_\phi\bm{\phi}(\bm{x})$ around a chosen equilibrium, with the nonlinearity bounded by a state-independent envelope $\|\bm{\phi}(\bm{x})\| \leq \Delta_{\max}$, the LMI framework above extends directly. Treating $\bm{d}$ and $\bm{\phi}$ as two independent norm-bounded uncertainties and applying the S-procedure \cite{boyd_linear_1994} with multipliers $\alpha > 0$ and $\mu > 0$, the invariance condition becomes the LMI
    \begin{equation}
    \label{eq:linearized_lmi}
    \begin{bmatrix}
    \Phi_{11} & \mathbf{B}_\phi & \mathbf{E}_d \\
    \mathbf{B}_\phi^T & -(\mu/\Delta_{\max}^2)\mathbf{I} & \mathbf{0} \\
    \mathbf{E}_d^T & \mathbf{0} & -\alpha \mathbf{I}
    \end{bmatrix} \preceq 0,
    \end{equation}
    where $\Phi_{11} = \mathbf{A}\mathbf{Q} + \mathbf{Q}\mathbf{A}^T + \mathbf{B}\mathbf{Y} + \mathbf{Y}^T\mathbf{B}^T + (\alpha + \mu)\mathbf{Q}$. This is the faithful two-multiplier embedding of a norm-bounded $\bm{\phi}$; it is strictly less conservative than a single-multiplier ``augmented disturbance'' formulation that lumps $\bm{d}$ and $\bm{\phi}$ at a common scale. While computationally tractable, the resulting feasibility region shrinks rapidly with the rotation envelope, since $\Delta_{\max}$ must be evaluated at the worst-case state in the constraint set. A rigorous, approximation-free extension using state-coupled sector-bounded nonlinearities is developed in Section~\ref{sec:quaternion_safety_filter}, and the two formulations are compared in Section~\ref{sec:conservatism_analysis}.
\end{remark}

\section{Extension to Nonlinear Systems: Quaternion Formulation}
\label{sec:quaternion_safety_filter}

This section presents a rigorous derivation of safety filters for attitude control using unit quaternion representation. Unlike Euler angles, quaternions avoid singularities and provide a global parameterization of $SO(3)$.

\subsection{Exact Quaternion Attitude Dynamics}

\subsubsection{Quaternion Kinematics}

A unit quaternion $\bm{q} = [q_0, \bm{q}_v^T]^T \in \mathbb{R}^4$ with scalar part $q_0 \in \mathbb{R}$ and vector part $\bm{q}_v = [q_1, q_2, q_3]^T \in \mathbb{R}^3$ satisfies the unit norm constraint:
\begin{equation}
\label{eq:quat_norm}
q_0^2 + \|\bm{q}_v\|^2 = 1
\end{equation}

The quaternion kinematics governing attitude evolution are:
\begin{equation}
\label{eq:quat_kinematics}
\dot{\bm{q}}_v = \frac{1}{2}(q_0 \mathbf{I}_3 + [\bm{q}_v]_\times)\bm{\omega}
\end{equation}
where $\bm{\omega} \in \mathbb{R}^3$ is the angular velocity in the body frame, $\mathbf{I}_3$ is the $3 \times 3$ identity matrix, and $[\bm{q}_v]_\times \in \mathbb{R}^{3 \times 3}$ is the skew-symmetric matrix:
\begin{equation}
[\bm{q}_v]_\times = \begin{bmatrix}
0 & -q_3 & q_2 \\
q_3 & 0 & -q_1 \\
-q_2 & q_1 & 0
\end{bmatrix}
\end{equation}

\subsubsection{Angular Velocity Dynamics}

The angular velocity dynamics for a rigid body are:
\begin{equation}
\label{eq:omega_dynamics}
\dot{\bm{\omega}} = \mathbf{J}^{-1}(-\bm{\omega} \times \mathbf{J}\bm{\omega} + \bm{u} + \bm{d})
\end{equation}
where $\mathbf{J} \in \mathbb{R}^{3 \times 3}$ is the positive definite inertia matrix, $\bm{u} \in \mathbb{R}^3$ is the control torque, and $\bm{d} \in \mathbb{R}^3$ is the external disturbance torque with $\|\bm{d}\| \leq d_{\max}$.

\subsection{State-Space Formulation and Linearization}

Define the reduced state vector using only the quaternion vector part:
\begin{equation}
\label{eq:state_vector}
\bm{z} = \begin{bmatrix} \bm{q}_v \\ \bm{\omega} \end{bmatrix} \in \mathbb{R}^6
\end{equation}

The scalar part $q_0$ is recovered from the unit norm constraint \eqref{eq:quat_norm} as $q_0 = \pm\sqrt{1 - \|\bm{q}_v\|^2}$. For the shortest rotation path, we choose $q_0 = +\sqrt{1 - \|\bm{q}_v\|^2}$.

At the hover equilibrium $(q_0, \bm{q}_v, \bm{\omega}) = (1, \mathbf{0}, \mathbf{0})$, the linearized system is:
\begin{equation}
\label{eq:linearized_system}
\dot{\bm{z}} = \mathbf{A}\bm{z} + \mathbf{B}\bm{u} + \bar{\mathbf{E}}\bar{\bm{d}}
\end{equation}
where:
\begin{equation}
\mathbf{A} = \begin{bmatrix}
\mathbf{0}_{3\times3} & \frac{1}{2}\mathbf{I}_3 \\
\mathbf{0}_{3\times3} & \mathbf{0}_{3\times3}
\end{bmatrix}, \quad
\mathbf{B} = \begin{bmatrix}
\mathbf{0}_{3\times3} \\
\mathbf{J}^{-1}
\end{bmatrix}, \quad
\bar{\mathbf{E}} = \begin{bmatrix}
\mathbf{0}_{3\times3} \\
d_{\max} \mathbf{J}^{-1}
\end{bmatrix}
\end{equation}
with normalized disturbance $\|\bar{\bm{d}}\| \leq 1$.

\subsection{Nonlinear System with Sector-Bounded Nonlinearity}

The exact nonlinear system can be written as:
\begin{equation}
\label{eq:nonlinear_system}
\dot{\bm{z}} = \mathbf{A}\bm{z} + \mathbf{B}\bm{u} + \mathbf{B}_\phi \bm{\phi}(\bm{z}) + \bar{\mathbf{E}}\bar{\bm{d}}
\end{equation}
where $\mathbf{B}_\phi = [\mathbf{I}_3^T,\, \mathbf{0}^T]^T$ and the nonlinearity $\bm{\phi}: \mathbb{R}^6 \to \mathbb{R}^3$ captures the deviation from linear kinematics:
\begin{equation}
\label{eq:nonlinearity}
\bm{\phi}(\bm{z}) = \bm{\phi}(\bm{q}_v, \bm{\omega}) = \frac{1}{2}[(q_0 - 1)\mathbf{I}_3 + [\bm{q}_v]_\times]\bm{\omega}
\end{equation}
with $q_0 = \sqrt{1 - \|\bm{q}_v\|^2}$.

\begin{remark}[Coriolis Term]
\label{rem:coriolis}
The model \eqref{eq:nonlinear_system} omits the Coriolis term $-\mathbf{J}^{-1}(\bm{\omega} \times \mathbf{J}\bm{\omega})$ present in \eqref{eq:omega_dynamics}. Near the hover equilibrium, this term is $\mathcal{O}(\|\bm{\omega}\|^2)$ and negligible for small angular rates. For rigorous safety guarantees, it can be cancelled exactly by augmenting the backup control input with a feedforward term:
\begin{equation}
\bm{u}_{\text{applied}} = \bm{u}_{\text{LMI}} + \bm{\omega} \times \mathbf{J}\bm{\omega}
\end{equation}
Since $\bm{\omega}$ is directly measured, this requires no additional sensors and incurs negligible computational cost. After this cancellation, \eqref{eq:nonlinear_system} describes the closed-loop dynamics exactly.
\end{remark}

\subsection{Rigorous Sector Bound Analysis}

To formulate tractable LMI conditions, we establish a rigorous sector bound relating $\phi(z)$ to the state $z$.

\begin{lemma}[Exact Sector Bound for Quaternion Nonlinearity]
\label{lem:sector_bound}
For all $\bm{q}_v \in \mathbb{R}^3$ satisfying $\|\bm{q}_v\| \leq \bar{q}_v$ and all $\bm{\omega} \in \mathbb{R}^3$, the nonlinearity $\bm{\phi}(\bm{q}_v, \bm{\omega})$ defined in \eqref{eq:nonlinearity} satisfies:
\begin{equation}
\label{eq:sector_bound}
\|\bm{\phi}(\bm{q}_v, \bm{\omega})\|^2 \leq \gamma^2 \|\bm{\omega}\|^2
\end{equation}
where:
\begin{equation}
\label{eq:gamma_exact}
\gamma^2 = \frac{1}{4}\left[(1 - \sqrt{1 - \bar{q}_v^2})^2 + \bar{q}_v^2\right] = \frac{1}{2}\left(1 - \sqrt{1 - \bar{q}_v^2}\right)
\end{equation}
\end{lemma}

\begin{proof}
From~\eqref{eq:nonlinearity}, $\bm{\phi} = \tfrac{1}{2}[(q_0-1)\bm{\omega} + \bm{q}_v\times\bm{\omega}]$.
Since $(q_0-1)\bm{\omega}$ is parallel to $\bm{\omega}$ and $\bm{q}_v\times\bm{\omega}$ is orthogonal to $\bm{\omega}$, the two terms are mutually orthogonal. By the Pythagorean theorem and the identity $\|\bm{q}_v\times\bm{\omega}\|^2 = \|\bm{q}_v\|^2\|\bm{\omega}\|^2 - (\bm{q}_v^T\bm{\omega})^2$:
\begin{equation}
    \|\bm{\phi}\|^2 \leq \tfrac{1}{4}\bigl[(q_0-1)^2 + \|\bm{q}_v\|^2\bigr]\|\bm{\omega}\|^2.
\end{equation}
Define $f(r) = (\sqrt{1-r^2}-1)^2 + r^2$ with $r = \|\bm{q}_v\|$. Since $f'(r) = 2r/\sqrt{1-r^2} > 0$ for $r > 0$, the maximum over $r \in [0,\bar{q}_v]$ is attained at $r = \bar{q}_v$. Expanding $f(\bar{q}_v)$ and simplifying gives $\gamma^2 = \tfrac{1}{4}f(\bar{q}_v) = \tfrac{1}{2}(1 - \sqrt{1-\bar{q}_v^2})$.
\end{proof}

\begin{corollary}
\label{cor:sector_bound_state}
The sector bound \eqref{eq:sector_bound} can be expressed in terms of the state vector $\bm{z}$ as:
\begin{equation}
\|\bm{\phi}(\bm{z})\|^2 \leq \gamma^2 \|\mathbf{E}_{\text{sec}} \bm{z}\|^2
\end{equation}
where $\mathbf{E}_{\text{sec}} = [\mathbf{0}_{3\times3},\ \mathbf{I}_3] \in \mathbb{R}^{3 \times 6}$ extracts the angular velocity from the state vector.
\end{corollary}

\subsection{LMI Formulation via S-Procedure}

Applying Corollary~\ref{cor:nagumo_ellipsoid} to system~\eqref{eq:nonlinear_system} with state feedback $\bm{u} = \mathbf{K}\bm{z}$, the set $\mathcal{E}(\mathbf{P})$ is RCI if and only if
\begin{equation}
\label{eq:nagumo_condition}
\bm{z}^T \mathbf{P} [\mathbf{A}\bm{z} + \mathbf{B}\mathbf{K}\bm{z} + \mathbf{B}_\phi\bm{\phi}(\bm{z}) + \bar{\mathbf{E}}\bar{\bm{d}}] \leq 0
\end{equation}
holds for all $\bm{z} \in \partial\mathcal{E}(\mathbf{P})$, all $\bm{\phi}$ satisfying $\|\bm{\phi}\|^2 \leq \gamma^2\|\mathbf{E}_{\text{sec}}\bm{z}\|^2$, and all $\bar{\bm{d}}$ satisfying $\|\bar{\bm{d}}\|^2 \leq 1$. To convert this semi-infinite constraint into a finite-dimensional LMI, we apply the S-procedure with multipliers for the boundary constraint ($\bm{z}^T\mathbf{P}\bm{z} = 1$) and the inequality constraints on $\bm{\phi}$ and $\bar{\bm{d}}$.

\begin{theorem}[LMI Formulation for Quaternion Safety Filter]
\label{thm:lmi_quaternion}
For fixed scalars $\alpha \in \mathbb{R}$ and $\tau > 0$, the ellipsoidal set $\mathcal{E}(\mathbf{P})$ is RCI for the quaternion-based system \eqref{eq:nonlinear_system} with state feedback control law $\bm{u} = \mathbf{K}\bm{z}$ if there exist matrices $\mathbf{Q} = \mathbf{P}^{-1} \succ 0$ and $\mathbf{Y}$ satisfying the following LMI in $(\mathbf{Q}, \mathbf{Y})$:
\begin{equation}
\label{eq:lmi_quaternion}
\begin{bmatrix}
\Psi_{11} & \mathbf{B}_\phi & \bar{\mathbf{E}} & \gamma\sqrt{\tau}\mathbf{Q}\mathbf{E}_{\text{sec}}^T \\
\mathbf{B}_\phi^T & -\tau \mathbf{I}_3 & \mathbf{0} & \mathbf{0} \\
\bar{\mathbf{E}}^T & \mathbf{0} & -\alpha \mathbf{I}_3 & \mathbf{0} \\
\gamma\sqrt{\tau}\mathbf{E}_{\text{sec}}\mathbf{Q} & \mathbf{0} & \mathbf{0} & -\mathbf{I}_3
\end{bmatrix} \preceq 0
\end{equation}
where $\Psi_{11} = \mathbf{A}\mathbf{Q} + \mathbf{Q}\mathbf{A}^T + \mathbf{B}\mathbf{Y} + \mathbf{Y}^T\mathbf{B}^T + \alpha\mathbf{Q}$. The feedback gain is recovered as $\mathbf{K} = \mathbf{Y}\mathbf{Q}^{-1}$.
\end{theorem}

\begin{proof}
Following the same S-procedure pattern as Theorem~\ref{thm:lmi_rfcis}, augment the state as $\bm{\xi} = [\bm{z}^T, \bm{\phi}^T, \bar{\bm{d}}^T]^T$ and apply Lemma~\ref{lem:s_procedure} to~\eqref{eq:nagumo_condition} with an unrestricted multiplier $\alpha$ (for the boundary equality $\bm{z}^T\mathbf{P}\bm{z}=1$), $\tau \geq 0$ (for the sector bound), and $\sigma = \alpha$ (for $\|\bar{\bm{d}}\|^2\leq 1$). After the congruence transformation $\mathrm{diag}(\mathbf{Q}, \mathbf{I}_3, \mathbf{I}_3)$ with $\mathbf{Q}=\mathbf{P}^{-1}$ and substitution $\mathbf{Y}=\mathbf{K}\mathbf{Q}$, the $(1,1)$ block contains the nonlinear term $\mathbf{Q}\mathbf{S}\mathbf{Q}$ with $\mathbf{S} = \tau\gamma^2\mathbf{E}_{\text{sec}}^T\mathbf{E}_{\text{sec}}$. Writing $\mathbf{Q}\mathbf{S}\mathbf{Q} = (\gamma\sqrt{\tau}\,\mathbf{Q}\mathbf{E}_{\text{sec}}^T)\mathbf{I}_3^{-1}(\gamma\sqrt{\tau}\,\mathbf{E}_{\text{sec}}\mathbf{Q})$ and applying the Schur complement lemma to linearize this quadratic term yields the $4\times4$ block LMI~\eqref{eq:lmi_quaternion}.
\end{proof}

\begin{remark}[Geometric Interpretation of State Constraint]
The constraint $\|\bm{q}_v\| \leq \bar{q}_v$ on the quaternion vector part corresponds to bounding the total rotation angle $\theta$ from the hover equilibrium via $\bar{q}_v = \sin(\theta_{\max}/2)$.
\end{remark}

\subsection{Complete Optimization Problem}

The quaternion-based robust safety filter synthesis is formulated as:
\begin{problem}[Quaternion Safety Filter Synthesis]
\label{prob:quaternion_synthesis}
For fixed scalars $\alpha \in \mathbb{R}$ and $\tau > 0$:
\begin{equation}
\begin{aligned}
\max_{\mathbf{Q}, \mathbf{Y}} \quad & \log\det(\mathbf{Q}) \\
\text{subject to:} \quad & \mathbf{Q} \succ 0 \\
& \text{LMI } \eqref{eq:lmi_quaternion} \\
& \text{Input constraints (Proposition \ref{prop:input_constraint}, $i=1,2,3$)} \\
& \begin{bmatrix}\mathbf{Q} & \mathbf{Q}\mathbf{C}_k^T \\ \mathbf{C}_k\mathbf{Q} & \bar{y}_k^2\mathbf{I}_3\end{bmatrix} \succeq 0,\; k=1,2
\end{aligned}
\end{equation}
where $\mathbf{C}_1 = [\mathbf{I}_3,\mathbf{0}]$, $\bar{y}_1 = \bar{q}_v$ and $\mathbf{C}_2 = [\mathbf{0},\mathbf{I}_3]$, $\bar{y}_2 = \omega_{\max}$ enforce $\|\bm{q}_v\| \leq \bar{q}_v$ and $\|\bm{\omega}\| \leq \omega_{\max}$ for all $\bm{z}\in\mathcal{E}(\mathbf{P})$.
\end{problem}

Upon solving this optimization problem, the ellipsoidal RCI set is $\mathcal{E}(\mathbf{P}^*)$ with $\mathbf{P}^* = (\mathbf{Q}^*)^{-1}$, and the backup control law is $\bm{u} = \mathbf{K}^*\bm{z}$ with $\mathbf{K}^* = \mathbf{Y}^*(\mathbf{Q}^*)^{-1}$.

\begin{remark}[Numerical Implementation]
For fixed $(\alpha, \tau)$, Problem~\ref{prob:quaternion_synthesis} is a convex SDP in $(\mathbf{Q}, \mathbf{Y})$ solvable by standard solvers (MOSEK, SeDuMi, SDPT3) via CVX or YALMIP. The optimal $(\alpha, \tau)$ are found by a 2D grid search over these parameters.
\end{remark}

\section{Numerical Results}\label{sec:numerical_results}
In this section, we demonstrate the proposed robust safety filter synthesis method on a quadrotor system subject to wind disturbances and input constraints. \footnote{Codes are available at: \url{https://github.com/Faivex/quaternion-safety-filter}}

\subsection{Quadrotor System Model}
Consider a quadrotor system with state vector $\bm{x} = [\bm{p}^T, \bm{v}^T, \bm{\eta}^T, \bm{\omega}^T]^T \in \mathbb{R}^{12}$, where:
\begin{itemize}
    \item $\bm{p} = [x, y, z]^T$: position in inertial frame (z-axis pointing downward)
    \item $\bm{v} = [v_x, v_y, v_z]^T$: linear velocity in inertial frame
    \item $\bm{\eta} = [\phi, \theta, \psi]^T$: roll, pitch, yaw angles
    \item $\bm{\omega} = [p, q, r]^T$: angular rates
\end{itemize}

The nonlinear dynamics can be written as:
\begin{align}
    \dot{\bm{p}} &= \bm{v} \\
    \dot{\bm{v}} &= g\bm{e}_3 - \frac{1}{m}\mathbf{R}(\bm{\eta})(F\bm{e}_3) \\
    \dot{\bm{\eta}} &= \mathbf{W}(\bm{\eta})\bm{\omega} \\
    \dot{\bm{\omega}} &= \mathbf{J}^{-1}(-\bm{\omega} \times \mathbf{J}\bm{\omega} + \bm{\tau} + \bm{d})
\end{align}
where $\mathbf{R}(\bm{\eta})$ is the rotation matrix from body to inertial frame, $\mathbf{W}(\bm{\eta})$ is the Euler angle rate transformation matrix, $F$ is the total thrust, $\bm{\tau}$ are the torque inputs, and $\bm{d}$ is the external disturbance. Numerical parameters and controller gains are listed in Table~\ref{tab:params} (Crazyflie~2.0~\cite{forster2015system}).

\begin{table}[!t]
\centering
\caption{Quadrotor parameters and controller gains.}
\label{tab:params}
\setlength{\tabcolsep}{4pt}
\renewcommand{\arraystretch}{1.1}
\begin{tabular}{llll}
\toprule
Parameter & Value & Parameter & Value \\
\midrule
$m$            & $0.028$\,kg                        & $g$      & $9.81$\,m/s$^2$ \\
$J_{xx}=J_{yy}$& $16.6\times10^{-6}$\,kg$\cdot$m$^2$ & $J_{zz}$ & $29.3\times10^{-6}$\,kg$\cdot$m$^2$ \\
\midrule
$k_{px}$             & $0.2$              & $k_{dx}$              & $0.2$ \\
$k_{py}$             & $-0.2$             & $k_{dy}$              & $-0.2$ \\
$k_{p\phi}=k_{p\theta}$ & $-1\times10^{-3}$ & $k_{d\phi}=k_{d\theta}$ & $2\times10^{-4}$ \\
$k_{p\psi}$          & $-3\times10^{-4}$  & $k_{d\psi}$           & $1\times10^{-4}$ \\
\bottomrule
\end{tabular}
\end{table}

We consider a hierarchical control structure with an outer-loop position controller and an inner-loop attitude controller.
The outer-loop controller computes desired thrust based on z position error and desired roll and pitch based on y and x position errors, respectively.
The inner-loop controller tracks attitude commands using desired torques $\tau_d$ (see Figure \ref{fig:control_structure}).

\begin{figure}[t]
    \centering
    \begin{tikzpicture}[
        auto,
        scale=0.75,
        transform shape,
        block/.style={rectangle, draw, thick, minimum width=1.5cm, minimum height=0.8cm, align=center, font=\small},
        sum/.style={circle, draw, thick, minimum size=0.4cm, inner sep=0pt, font=\small},
        >=latex]
        \node[coordinate] (input) {};
        \node[sum, right=0.5cm of input] (sum1) {$-$};
        
        \node[block, right=0.5cm of sum1, minimum height=1.5cm] (pos_ctrl) {Position\\ Controller};
        \draw[->] (sum1) -- node[font=\small]{$\mathbf{e}_p$} (pos_ctrl);
        \node [draw, thick, minimum width=0.5cm, minimum height=0.6cm, right=1.2cm of pos_ctrl, yshift = -0.3cm] (saturation) 
        {\begin{tikzpicture}[scale=0.25, transform shape]
            \draw[-] (1,1) -- (0.3,1);
            \draw[-] (0.3,1) -- (-0.3,-1);
            \draw[-] (-0.3,-1) -- (-1,-1);
            \end{tikzpicture}};
        \draw[->] (pos_ctrl.east) |- node[font=\small, xshift=0.6cm, yshift=0.0cm]{$[\phi_d, \theta_d]$} (saturation.west);
        \node[block, right=1.2cm of saturation] (att_ctrl) {Attitude\\ Controller};
        \draw[->] (saturation.east) -- node[font=\small]{$[\bar{\phi}_d, \bar{\theta}_d]$} (att_ctrl.west);
        \node[block, right=0.55cm of att_ctrl, minimum height=1.0cm, minimum width=0.5cm, text=blue!80!black, draw=blue!80!black] (safety_filter) {Robust \\ Safety\\ Filter};
        \draw[->] (att_ctrl.east) -- node[font=\small]{$\bar{\boldsymbol{\tau}}_d$} (safety_filter.west);
        \node[block, right=7.0cm of pos_ctrl, minimum height=1.5cm] (quad) {Quadrotor\\ Dynamics};
        \draw[->] (pos_ctrl.east) |- ($(quad.west) + (0,0.4)$) node[font=\small, xshift=-3.3cm, yshift = 0.2cm]{$F$};
        \draw[->] (safety_filter.east) |- ($(quad.west) + (0,-0.3)$)node[font=\small, xshift = -0.3cm, yshift = 0.2cm]{$\boldsymbol{\tau}_s$};
        \node[coordinate, left=0.5cm of att_ctrl, yshift = -0.3cm] (psi_in) {};
        \draw[->] (psi_in) |- ($(att_ctrl.west) + (0,-0.3)$) node[font=\small,  xshift=-0.7cm]{$\psi_d$} (att_ctrl); 
        \draw[->] (input) -- node[font=\small]{$\boldsymbol{p}_d$} (sum1);
        \draw[->] (quad.south) |- ($(sum1.south) + (0,-1.2)$) -| (sum1.south) node[font=\small, xshift = 0.2cm, yshift = -0.2cm]{$\boldsymbol{p}$};
        \draw[->] (quad.south) |- ($(sum1.south) + (0,-1.2)$) -| (att_ctrl.south) node[font=\small, xshift = 0.4cm, yshift = -0.4cm]{$[\boldsymbol{\eta}, \boldsymbol{\omega}]$};
        \draw[->] (quad.south) |- ($(sum1.south) + (0,-1.2)$) -| (safety_filter.south) node[font=\small, xshift = 0.45cm, yshift = -0.3cm]{$[\boldsymbol{\eta}, \boldsymbol{\omega}]$};
        \draw[->] ($(quad.north)+(0,0.4)$) -- node[right, font=\small]{$\boldsymbol{d}(t)$} (quad.north);
    \end{tikzpicture}
    \caption{Block diagram of the hierarchical control structure for the quadrotor. 
    The outer loop controls position $p$ by generating desired thrust $F_d$, roll, and pitch commands, 
    while the inner loop tracks these commands using thrust $F$ and torques $\tau$. External disturbances $d(t)$ affect the quadrotor dynamics as external torques. 
    The robust safety filter modifies the clipped desired torque commands $\bar{\boldsymbol{\tau}}_d$ to ensure safety.}
    \label{fig:control_structure}
\end{figure}

Here, altitude ($z$) is regulated by a separate independent loop that maintains constant hover thrust; only the horizontal position errors generate the roll and pitch attitude references that are passed to the safety filter.
The outer-loop position controller is a PD controller:
\begin{align}
    \phi_d &= k_{py}(y_d - y) + k_{dy}(0 - v_y) \\
    \theta_d &= k_{px}(x_d - x) + k_{dx}(0 - v_x)
\end{align}

These desired roll and pitch angles are then saturated to ensure they remain within feasible limits:
\begin{align}
    \bar{\phi}_d &= \max(\min(\phi_d, \phi_{\max}), -\phi_{\max}) \\
    \bar{\theta}_d &= \max(\min(\theta_d, \theta_{\max}), -\theta_{\max})
\end{align}
where $\phi_{\max} = \theta_{\max} = 60^\circ$.

The inner-loop attitude controller is also a PD controller:
\begin{align}
    \tau_{d,x} &= k_{p\phi}(\phi - \phi_d) + k_{d\phi}(0 - p) \\
    \tau_{d,y} &= k_{p\theta}(\theta - \theta_d) + k_{d\theta}(0 - q) \\
    \tau_{d,z} &= k_{p\psi}(\psi - \psi_d) + k_{d\psi}(0 - r)
\end{align}
with gains as listed in Table~\ref{tab:params}.

The control inputs are then clipped ($\bar{\boldsymbol{\tau}}_d$) to ensure they remain within feasible limits before being sent to the safety filter.

\begin{figure*}[!t]
    \centering
    \includegraphics[width=\textwidth]{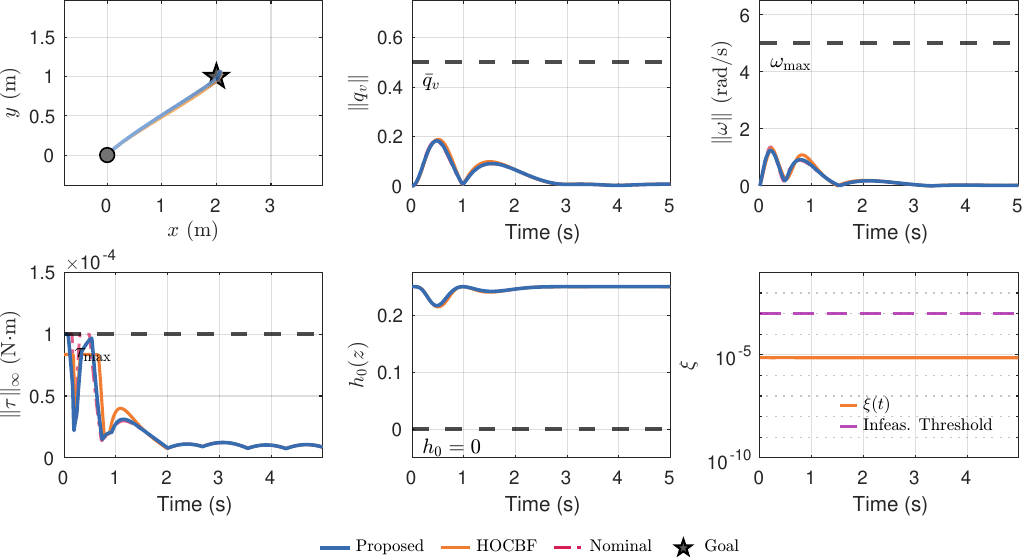}
    \caption{Scenario~I (set-point with small initial errors): three-way comparison of the proposed filter, the HOCBF baseline, and the unfiltered nominal controller. Top row: 2D trajectory, $\|\bm{q}_v\|$ vs.\ $\bar{q}_v$, $\|\bm{\omega}\|$ vs.\ $\omega_{\max}$. Bottom row: $\|\bm{\tau}\|_\infty$ vs.\ $\tau_{\max}$, HOCBF barrier $h_0(\bm{z})$ (positive values indicate safety), and HOCBF slack $\xi(t)$. All three methods remain safely inside the constraint set throughout.}
    \label{fig:scenario_I}
\end{figure*}

\begin{figure*}[!t]
    \centering
    \includegraphics[width=\textwidth]{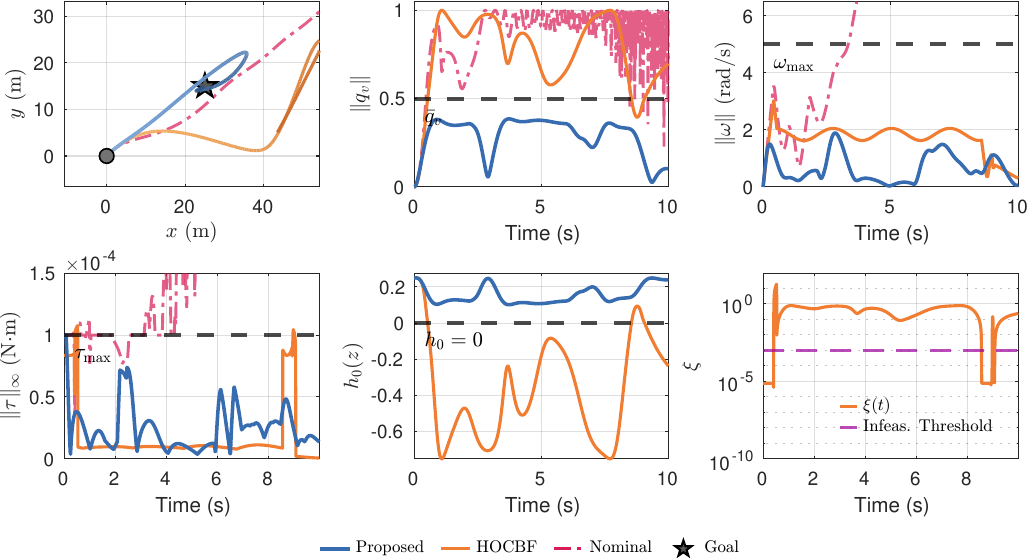}
    \caption{Scenario~II (set-point with large initial errors): three-way comparison. The outer-loop position controller saturates the inner-loop attitude command, driving both the HOCBF baseline and the unfiltered system across the safety boundary. The proposed filter maintains $\|\bm{q}_v\| \leq \bar{q}_v$ throughout. The slack panel shows persistent HOCBF infeasibility.}
    \label{fig:scenario_II}
\end{figure*}

\begin{figure*}[!t]
    \centering
    \includegraphics[width=\textwidth]{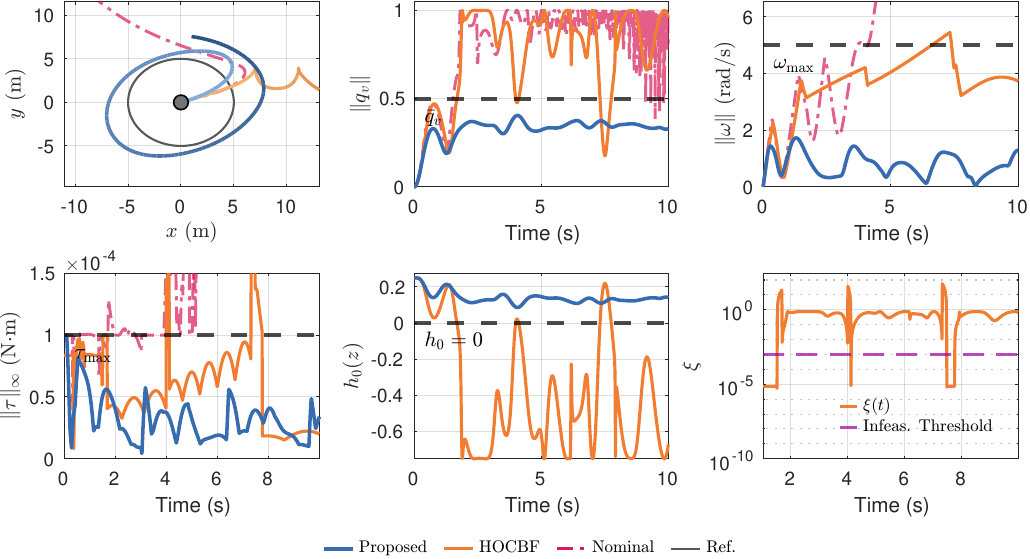}
    \caption{Scenario~III (high-frequency circular tracking, $r=5$\,m, $\dot{\psi}=1$\,rad/s): three-way comparison. Persistent excitation of the attitude dynamics drives the HOCBF and the unfiltered system to constraint violation; the proposed filter respects the safety bound at every instant.}
    \label{fig:scenario_III}
\end{figure*}

\subsection{Safety Filter Implementation}
Such a control structure can lead to unsafe behavior if the position errors become too large.
Large position errors can result in excessive roll and pitch commands, which may saturate the attitude controller and lead to instability.
To mitigate this, we implement the robust safety filter from Section~\ref{sec:quaternion_safety_filter} that constrains the attitude to safe bounds, ensuring stability even under external disturbances.

The safety state is $\bm{z} = [\bm{q}_v^T,\, \bm{\omega}^T]^T \in \mathbb{R}^6$, governed by~\eqref{eq:nonlinear_system} with the system matrices and sector-bound gain $\gamma$ as defined in Section~\ref{sec:quaternion_safety_filter}, and with the control input identified as the quadrotor torque ($\bm{u} \equiv \bm{\tau}$).

The attitude safety constraint is expressed in quaternion coordinates as
\begin{equation}
    \|\bm{q}_v\| \leq \bar{q}_v = \sin\!\left(\frac{\theta_{\max}}{2}\right) = 0.500,\quad \theta_{\max} = 60^\circ,
\end{equation}
together with $\|\bm{\omega}\|_2 \leq \omega_{\max} = 5$\,rad/s.
The disturbance bound is $d_{\max} = 1\times10^{-5}$\,N$\cdot$m and the per-axis torque limits are
$|\tau_i| \leq \tau_{\max} = 1\times10^{-4}$\,N$\cdot$m, $i=1,2,3$.

The RCI ellipsoid $\mathcal{E}(\mathbf{P}^*)$ and backup controller $\bm{\tau}_b = \mathbf{K}^*\bm{z}$ are obtained by solving Problem~\ref{prob:quaternion_synthesis}.
Since Problem~\ref{prob:quaternion_synthesis} is a convex SDP for fixed $(\alpha,\tau)$, we perform a $5\times5$ grid search with both parameters drawn from $\{10^{-1},\,10^{-0.25},\,10^{0.5},\,10^{1.25},\,10^{2}\}$, selecting the pair that maximises $\log\det\mathbf{Q}$.

The safety filter interpolates between the backup controller and the nominal controller according to \eqref{eq:safety_filter_law}, with Lyapunov function $V(\bm{z}) = \bm{z}^T\mathbf{P}^*\bm{z}$ and mixing coefficient defined by \eqref{eq:mixing_function} using $h_{\min} = 0.1$ and $h_{\max} = 0.9$.
The RCI boundary is at $V(\bm{z}) = 1$, and the filter is inactive ($\alpha = 0$) when $V(\bm{z}) \leq h_{\min}$.
To compensate for the gyroscopic term (Remark~\ref{rem:coriolis}), the applied torque is augmented as $\bm{\tau}_{\text{applied}} = \bm{\tau} + \bm{\omega}\times\mathbf{J}\bm{\omega}$.

\subsection{Simulation Results}
We evaluate the performance of the robust safety filter through numerical simulations in MATLAB
using the CVX toolbox \cite{cvx} with the MOSEK solver \cite{mosek}.
Three scenarios are considered:
\begin{itemize}
    \item \textbf{Scenario~I}: set-point tracking with a small goal at $(x,y)=(2,1)$\,m over $T=5$\,s.
    \item \textbf{Scenario~II}: set-point tracking with a distant goal at $(x,y)=(25,15)$\,m over $T=10$\,s, producing large initial position errors.
    \item \textbf{Scenario~III}: circular trajectory with radius $5$\,m and angular speed $1$\,rad/s over $T=10$\,s.
\end{itemize}

The disturbance is modeled as $\bar{\bm{d}}(t) = [\sin(2t),\, \cos(2t),\, 0]^T$ and is active in all scenarios. Each scenario figure (Figs.~\ref{fig:scenario_I}--\ref{fig:scenario_III}) presents a three-way comparison among the proposed filter (deep blue), the HOCBF baseline (orange), and the unfiltered nominal controller (red), under identical disturbance and initial-condition profiles. Six panels report, top-to-bottom and left-to-right: the 2D position trajectory with time-graded colour, the safety-relevant norms $\|\bm{q}_v\|$ and $\|\bm{\omega}\|$ versus their bounds, the worst-axis torque $\|\bm{\tau}\|_\infty$ versus $\tau_{\max}$, the HOCBF barrier $h_0(\bm{z}) = \bar{q}_v^2 - \|\bm{q}_v\|^2$ (positive when safe), and the HOCBF slack variable $\xi(t)$ with infeasibility threshold. The HOCBF baseline used for comparison is formally defined in Section~\ref{sec:cbf_comparison}; the following describes what each figure shows.

Figure~\ref{fig:scenario_I} shows Scenario~I. The position error is small, so the inner-loop attitude demand stays well within the safe region for all three methods. Both the proposed filter and the HOCBF remain effectively inactive ($\xi(t)$ at the numerical noise floor), and all three trajectories converge to the goal. The HOCBF and the proposed filter are indistinguishable from the nominal controller in the position panel, confirming that neither filter causes meaningful deviation when intervention is not required.

Figure~\ref{fig:scenario_II} shows Scenario~II. Large initial position errors saturate the outer-loop position controller, driving the commanded attitude into a regime where the HOCBF condition cannot be satisfied within the actuator bounds. The slack panel shows $\xi(t)$ above the threshold throughout the maneuver, the barrier $h_0$ becomes negative within the first second, and $\|\bm{q}_v\|$ saturates near unity. The unfiltered nominal controller exhibits the same failure mode. The proposed filter maintains $\|\bm{q}_v\| \leq \bar{q}_v$ at every instant and converges to the setpoint.

Figure~\ref{fig:scenario_III} shows Scenario~III. The aggressive circular reference continuously excites the attitude dynamics. Without intervention, the unfiltered system becomes unstable; the HOCBF baseline first violates the constraint at $t = 1.69$\,s and remains in the unsafe set thereafter, with $\xi(t)$ persistently active. The proposed filter respects the safety bound throughout the maneuver while tracking the circular reference.

\subsection{Comparison with a Higher-Order CBF Baseline}
\label{sec:cbf_comparison}

To position the proposed filter against the dominant safety-filter paradigm, we compare against a Higher-Order Control Barrier Function (HOCBF) baseline~\cite{xiao2019control,nguyen2016exponential} solved as a QP~\cite{ames2016control,ames_control_2019}. The natural barrier $h_0(\bm{z}) = \bar{q}_v^2 - \|\bm{q}_v\|^2$ has relative degree two, so we introduce $h_1 = \dot{h}_0 + \kappa_0 h_0$ and enforce $\dot{h}_1 + \kappa_1 h_1 \geq 0$. Using the identities $\bm{q}_v^T([\bm{q}_v]_\times\bm{\omega}) = 0$ and $\bm{\omega}^T([\bm{q}_v]_\times\bm{\omega}) = 0$, the condition reduces to the slack-relaxed QP~\cite{ames2016control}:
\begin{equation}
\label{eq:cbf_qp}
\begin{aligned}
\min_{\bm{u},\,\xi}\;\;& \tfrac{1}{2}\|\bm{u}-\bm{u}_{\text{nom}}\|^2 + \tfrac{p}{2}\xi^2 \\
\text{s.t.}\;\;& L_g h_1\,\bm{u} + \xi \geq -\,\mathrm{drift}_{h_1}(\bm{z}) - \kappa_1\,h_1(\bm{z}), \\
& -\tau_{\max}\bm{1} \leq \bm{u} \leq \tau_{\max}\bm{1},\quad \xi \geq 0,
\end{aligned}
\end{equation}
where $\mathrm{drift}_{h_1} = \tfrac{1}{2}(\bm{q}_v^T\bm{\omega})^2 - \tfrac{1}{2}q_0^2\|\bm{\omega}\|^2 + \kappa_0\dot{h}_0$ and $L_g h_1 = -q_0\bm{q}_v^T\mathbf{J}^{-1}$, with $\kappa_0=20$, $\kappa_1=5$, $p=10^2$. Disturbances are omitted since this is the standard non-robust baseline. Outside the safe set ($h_0 < 0$), the QP is replaced by a saturated proportional braking law.

Table~\ref{tab:cbf_comparison} reports the comparison across all three
scenarios. In Scenario~I both methods satisfy the safety constraint. In Scenarios~II
and~III, the position-loop saturation drives the inner-loop attitude demand
into a regime where the HOCBF cannot satisfy its safety condition within
the actuator limits: the slack variable $\xi$ is active on $91\%$ and
$82\%$ of timesteps, the constraint is first violated at $t = 0.54$\,s and
$t = 1.69$\,s respectively, and $\|\bm{q}_v\|$ rapidly saturates near unity
(a full $180^\circ$ rotation). The proposed filter maintains
$\|\bm{q}_v\| \leq \bar{q}_v$ throughout all three scenarios.

\begin{table}[!t]
\centering
\caption{Comparison with the HOCBF baseline. The HOCBF is feasible only in
Scenario~I, where neither filter needs to intervene; it loses feasibility
under input saturation in Scenarios~II and~III, while the proposed filter
maintains safety in all cases.}
\label{tab:cbf_comparison}
\setlength{\tabcolsep}{4pt}
\renewcommand{\arraystretch}{1.1}
\begin{tabular}{l c c c}
\toprule
Metric & Scen.~I & Scen.~II & Scen.~III \\
\midrule
$\|\mathbf{q}_v\|_{\max}$, Proposed   & 0.182 & 0.385 & 0.404 \\
$\|\mathbf{q}_v\|_{\max}$, HOCBF      & 0.188 & 1.000 & 1.000 \\
$\|\mathbf{q}_v\|_{\max}$, Nominal    & 0.183 & 1.000 & 1.000 \\
Constraint $\bar{q}_v$                & 0.500 & 0.500 & 0.500 \\
\midrule
HOCBF slack-active (\% steps)         & 0\,\% & 91\,\% & 82\,\% \\
HOCBF first violation time            & --    & 0.54\,s & 1.69\,s \\
Proposed maintains safety             & \checkmark & \checkmark & \checkmark \\
HOCBF maintains safety                & \checkmark & $\times$ & $\times$ \\
\bottomrule
\end{tabular}
\end{table}

These results isolate the structural advantage of the LMI synthesis: the
LMI of Proposition~\ref{prop:input_constraint} binds the backup gain
$\mathbf{K}$ such that $|\mathbf{K}\bm{z}| \leq \tau_{\max}$ for all
$\bm{z} \in \mathcal{E}(\mathbf{P})$, by construction. The HOCBF baseline
has no analogous guarantee, and the slack relaxation quantifies the
resulting shortfall directly in Figs.~\ref{fig:scenario_II}
and~\ref{fig:scenario_III}: as $\xi(t)$ rises above its threshold, the
HOCBF trajectory crosses $\bar{q}_v$ and the barrier $h_0$ becomes
negative, while the proposed filter maintains $h_0 \geq 0$ throughout.

\subsection{Conservatism Analysis: Sector Bound vs.\ Linearization}
\label{sec:conservatism_analysis}

A separate question is how much the exact sector bound of
Lemma~\ref{lem:sector_bound} reduces conservatism relative to the
linearization-based formulation of Remark~\ref{rem:nonlinear_extension},
which treats $\bm{\phi}$ as a state-independent uncertainty with
$\|\bm{\phi}\| \leq \Delta_{\max} = \gamma\,\omega_{\max}$. We quantify
the gap by solving both LMIs over a two-dimensional grid in design space
$(\theta_{\max}, \omega_{\max})$, with $\theta_{\max} \in [5^\circ, 60^\circ]$
in $5^\circ$ steps and $\omega_{\max} \in [1, 10]$\,rad/s in $1$\,rad/s
steps (a total of 120 grid points). At each point we record whether each
formulation is feasible and, when feasible, the certified ellipsoid volume
$\mathrm{vol}(\mathcal{E}(\mathbf{P})) = \pi^{n/2}/\Gamma(n/2+1)\sqrt{\det\mathbf{Q}}$
with $n = 6$.

Figure~\ref{fig:conservatism_2D} summarises the result. The left panel
shows the RCI coverage of the constraint box for the proposed method,
$\mathrm{vol}(\mathcal{E}(\mathbf{P}^*)) / [(2\bar{q}_v)^3 (2\omega_{\max})^3]$,
which is feasible across the entire design space. The right panel shows
the multiplicative volume gain
$\mathrm{vol}(\text{Proposed})/\mathrm{vol}(\text{Linearized})$ over the
overlap region where both formulations are feasible.

\begin{figure}[!t]
\centering
\includegraphics[width=\columnwidth]{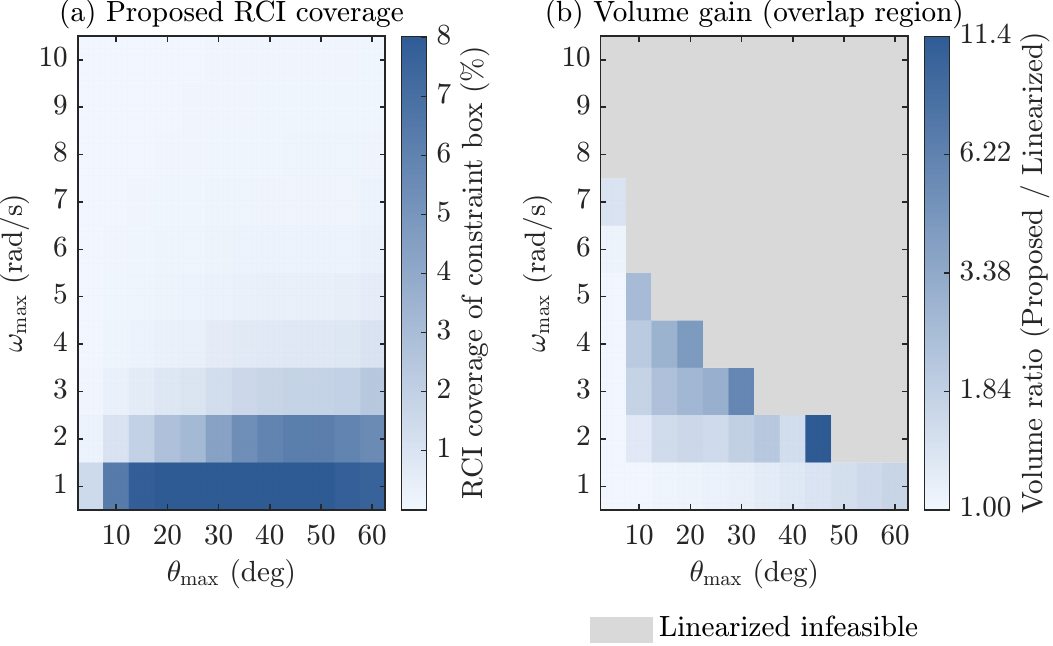}
\caption{Conservatism comparison over the $(\theta_{\max}, \omega_{\max})$
design space. \textbf{(a)} RCI coverage of the constraint box for the
proposed sector-bound LMI of Theorem~\ref{thm:lmi_quaternion};
the proposed method is feasible across the entire tested envelope.
\textbf{(b)} Volume gain
$\mathrm{vol}(\text{Proposed})/\mathrm{vol}(\text{Linearized})$ over the
overlap region where both LMIs are feasible. Grey cells indicate
linearized infeasibility. Colour scale is logarithmic; tick labels show
real multipliers.}
\label{fig:conservatism_2D}
\end{figure}

The proposed sector-bound LMI is feasible at all 120 grid points. The
linearization-based formulation is feasible at only 35 grid points
($29\%$), clustered at low $\theta_{\max}$ and low $\omega_{\max}$. As
either bound grows, the worst-case envelope
$\Delta_{\max} = \gamma\,\omega_{\max}$ exceeds the level the linearized
LMI can tolerate within the constraint set, and no feasible $(\mathbf{Q},
\mathbf{Y})$ exists. Where both formulations are feasible, the proposed
method produces a substantially larger RCI ellipsoid, with a median
volume gain of $1.49\times$ and a maximum of $11.4\times$.

\section{Conclusion}\label{sec:conclusion}
This paper presented an ellipsoidal set-theoretic framework for robust safety filter synthesis in constrained linear systems. The proposed LMI-based design simultaneously computes a maximal ellipsoidal RCI set and its associated state-feedback control law, ensuring safety by minimally modifying nominal control inputs only when necessary. The framework was extended to nonlinear attitude dynamics via a rigorous quaternion-based formulation: exact sector bounds on the quaternion kinematic nonlinearity were derived and incorporated into the LMI synthesis without approximation, yielding formal safety guarantees beyond the reach of linearization-based approaches. Numerical simulations on a quadrotor system validated the filter's effectiveness under external disturbances, large initial errors, and high-frequency trajectory tracking. A comparison against a standard HOCBF baseline showed that the proposed filter maintains safety in all scenarios while the HOCBF baseline loses feasibility under input saturation. A conservatism analysis over a 2D design-space sweep established that the exact sector bound enables certified synthesis across the entire tested envelope, whereas the faithful linearization-based formulation of Remark~\ref{rem:nonlinear_extension} is feasible on only $29\%$ of the same envelope. Future work may explore adaptive disturbance estimation and real-time hardware implementation of the safety filter.

\bibliographystyle{IEEEtran}
\bibliography{refList}


\begin{IEEEbiography}[{\includegraphics[width=1in,height=1.25in,clip,keepaspectratio]{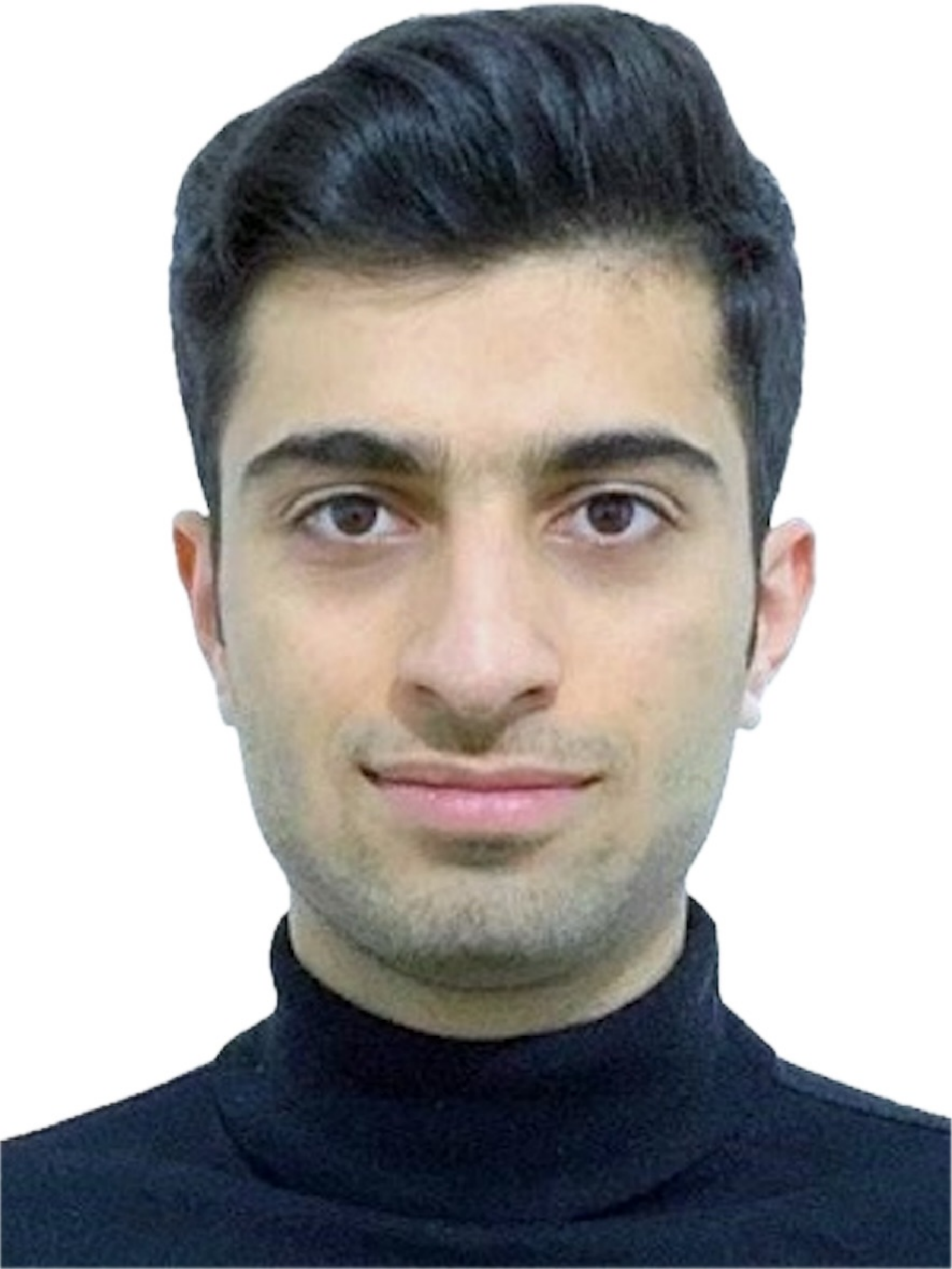}}]{Reza Pordal} completed his B.Sc.\ and M.Sc.\ degrees in Aerospace Engineering at Sharif University of Technology, Tehran, Iran, in 2021 and 2024, respectively. He received the Iranian Best B.Sc.\ Thesis Award from the Iranian Aerospace Society in 2022. He is currently a Research Assistant at the CNAV Laboratory, Sharif University of Technology. His research interests include safe and robust control of uncertain systems, safety-critical aerospace applications, and robotics for environmental purposes.
\end{IEEEbiography}

\begin{IEEEbiography}[{\includegraphics[width=1in,height=1.25in,clip,keepaspectratio]{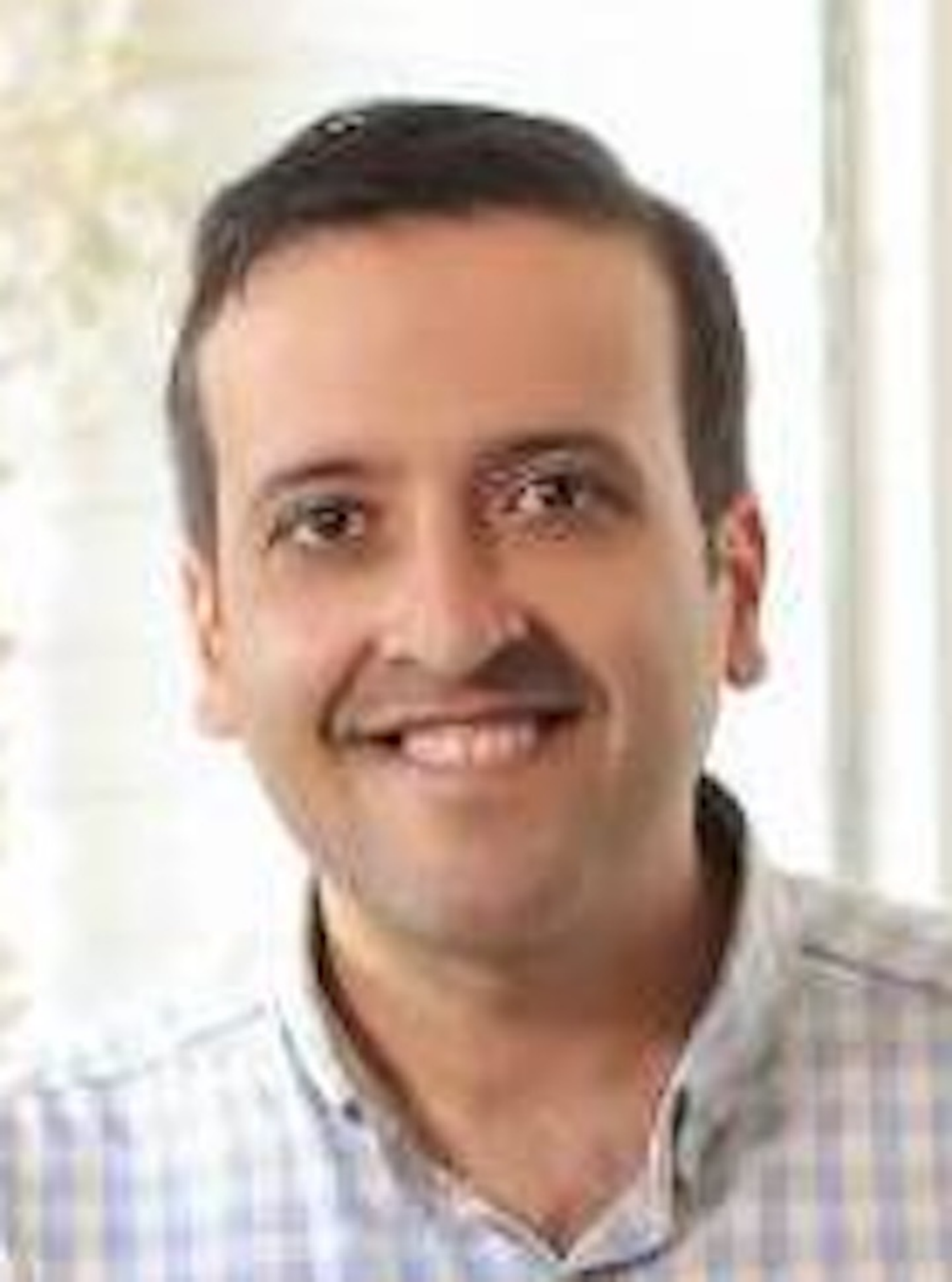}}]{Alireza Sharifi} received his M.Sc., and Ph.D. degrees in Aerospace Engineering from Sharif University of Technology, Tehran, Iran. He is currently an Assistant Professor in the Department of Aerospace Engineering at Sharif University of Technology. His research interests include air navigation, cooperative navigation, flight dynamics, and control theory. He has received multiple awards, including the Iranian Best PhD Thesis Award of the Year and the Outstanding Young Teacher Award from Sharif University of Technology.
\end{IEEEbiography}

\begin{IEEEbiography}[{\includegraphics[width=1in,height=1.25in,clip,keepaspectratio]{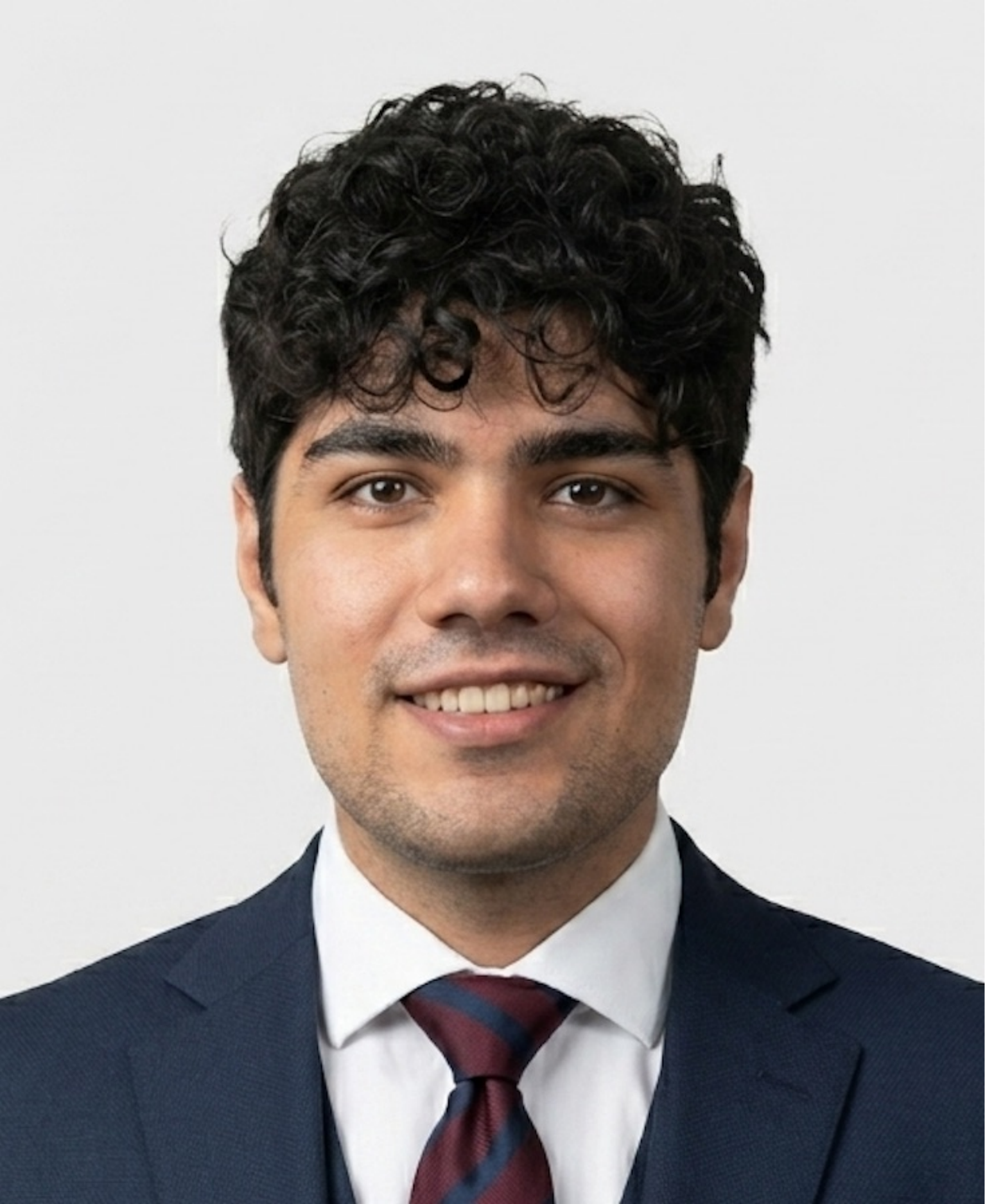}}]{Ali BaniAsad}
 received the B.Sc. degree in Aerospace Engineering from Sharif University of Technology, Tehran, Iran, in 2022, where he was the recipient of the Best Undergraduate Thesis Award. He received the M.Sc. degree in Aerospace Engineering from the same institution in 2025, with a dissertation on embedded reinforcement learning for robust control of robotic systems under adversarial disturbances. From 2024 to 2025, he was a Research and Development Engineer with Fasta Robotics, where his work concerned the design of autopilot control architectures and the integration of learning-based methods with classical guidance, navigation, and control for robots.
His research interests include reinforcement learning, robust and optimal control, multi-agent systems, differential games, and the deployment of learning-based controllers on embedded platforms. He has authored several peer-reviewed publications in these areas.
\end{IEEEbiography}

\vfill

\end{document}